\documentclass[]{interact}

\usepackage{epstopdf}
\usepackage[caption=false]{subfig}

\usepackage[numbers,sort&compress]{natbib}
\bibpunct[, ]{[}{]}{,}{n}{,}{,}
\renewcommand\bibfont{\fontsize{10}{12}\selectfont}
\makeatletter
\def\NAT@def@citea{\def\@citea{\NAT@separator}}
\makeatother

\theoremstyle{plain}
\newtheorem{theorem}{Theorem}[section]
\newtheorem{lemma}[theorem]{Lemma}
\newtheorem{corollary}[theorem]{Corollary}

\newtheorem{observation}[theorem]{Observation}
\newtheorem{claim}[theorem]{Claim}

\theoremstyle{definition}
\newtheorem{definition}[theorem]{Definition}

\theoremstyle{remark}

\usepackage{amsmath,amssymb,color,array,hyperref}
\usepackage[linesnumbered,vlined,ruled]{algorithm2e}
\begin{document}


\title{Circle Formation by Asynchronous Opaque Fat Robots on an Infinite Grid}

\author{
\name{Pritam Goswami\textsuperscript{a}\thanks{Pritam Goswami, Email: pritamgoswami.math.rs@jadavpuruniversity.in}, Manash Kumar Kundu \textsuperscript{b}, Satakshi Ghosh\textsuperscript{a} and Buddhadeb Sau\textsuperscript{a}}
\affil{\textsuperscript{a}Department of Mathematics,Jadavpur University, Kolkata-700032, India; \textsuperscript{b} Gayeshpur Government Polytechnic, Department of Science and Humanities, Kalyani, West Bengal - 741234, India}
}

\maketitle

\begin{abstract}
This study addresses the problem of "Circle Formation on an Infinite Grid by Fat Robots" ($CF\_FAT\_GRID$). Unlike prior work focused solely on point robots in discrete domain, it introduces fat robots to circle formation on an infinite grid, aligning with practicality as even small robots inherently possess dimensions. The algorithm, named $CIRCLE\_FG$, resolves the $CF\_FAT\_GRID$ problem using a swarm of fat luminous robots. Operating under an asynchronous scheduler, it achieves this with five distinct colors and by leveraging one-axis agreement among the robots.
\end{abstract}


\section{Introduction}

Distributed algorithm has nowadays become a very emerging and interesting topic. The reasons behind this are very clear. Distributed systems offer a multitude of advantages over centralized systems, making them a preferred choice for various applications. One key benefit is enhanced reliability and fault tolerance.
Additionally, distributed systems often exhibit superior scalability by allowing for easy expansion through the addition of new nodes, thus accommodating growing workloads and user demands seamlessly. This work focuses on swarm robot algorithm.
A system consisting of swarm robots is an example of a distributed system. A robot in a swarm is a small computing unit with the capability to move. These robots are considered to be autonomous (i.e., there is no central control), anonymous (i.e., the robots do not have any unique identifiers), homogeneous (i.e., all robots execute the same distributed algorithm), and identical (i.e., all robots are physically indistinguishable). In a swarm robot system, a collection of such robots are considered on an environment (i.e., euclidean plane, circle, discrete network, etc.) and their aim is to execute some tasks (e.g. Gathering \cite{AgmonP06,CiceroneSN18,GoswamiSGS22}, Dispersion \cite{MollaM22,KaurM23}, Exploration \cite{BramasDL20,DarwichUBLDL22}, Pattern Formation\cite{FlocchiniPSW08,GhoshGSS23} etc.) for which the mobility of the robots are needed. Swarm robotics has a huge application in many different scenarios such as patrolling, area coverage, network maintenance, etc.

Now one main direction of research in this field is to examine the optimal model of a robot in terms of different parameters. Memory and communication is one such parameter. There are mainly four robot models in the literature based on memory and communication. These four models are,
\begin{itemize}
    \item $\mathcal{OBLOT}$ Model: In this model, the robots do not have any persistent memory and there is no means of explicit communication between any two robots in the system.
    \item $\mathcal{FSTA}$ Model: In this robot model, the robots do not have any means of explicit communication between themselves. However, the robots have finite bits of persistent memory to remember their previous states.
    \item $\mathcal{FCOM}$ Model: In this model, the robots do not have any persistent memory but they can explicitly communicate with other robots using finite bit messages.
    \item $\mathcal{LUMI}$ Model: In this model, the robots have finite persistent memory and also can communicate with other robots using finite bit messages.
\end{itemize}
In every robot variant (excluding the $\mathcal{OBLOT}$ model), each individual robot comes equipped with a light that possesses a finite range of colors. These colors serve as a method of both communication and memory among the robots. In the $\mathcal{FSTA}$ model, communication is absent, as robots cannot perceive the lights of other robots. However, a robot can observe its own light, which functions as its personal memory. In the $\mathcal{FCOM}$ model, robots lack the ability to view their own light, yet they can perceive the lights of their counterparts. Lastly, in the $\mathcal{LUMI}$ model, a robot has the capacity to see both its own light and the lights of other robots, allowing for comprehensive communication and memory to remember finitely many previous states. 

Another parameter is robot vision. Categorized by visibility, there exist two distinct classifications of robot models: the \textit{Non-Restricted Visibility Model} and the \textit{Restricted Visibility Model}. In the Non-Restricted Visibility Model, a robot possesses the ability to observe the entire surroundings without any hindrance caused by other robots, and although extensively employed in existing literature (\cite{AgmonP06,FlocchiniPSW08,GhoshGSS23}) its practical feasibility is limited. Hence, the Restricted Visibility Model emerges. This model encompasses two potential visibility constraints. Firstly, a robot's visual range might be constrained, allowing it to only perceive objects within a certain distance, rather than the entire environment (\cite{GoswamiSGS22,FlocchiniPSW05,LunaUVY20}). Furthermore, visibility might be obstructed by the presence of other robots in the vicinity (\cite{Adhikary21,app13137991,inbook}). If robots do not interfere with the visibility of other robots, then these types of robots are called  \textit{Transparent Robots}, otherwise, the robots are called \textit{opaque Robots}. 

Another aspect of the robot model considers the dimension of the robots. This perspective divides robots into two primary classifications. These two divisions encompass robots without physical dimensions and those possessing specific dimensions. In the dimensionless robot model, a robot is conceptualized as a mere point in space. Conversely, within the model involving robots with dimensions, these entities are represented as disks with some radius. While the point robot model prevails in academic discourse (\cite{tanaka,AgmonP06,FlocchiniPSW05,FlocchiniPSW08}), it lacks suitability for real-world implementations. On the contrary, robots with dimensions, often referred to as \textit{Fat robots}, prove more pragmatic for practical applications(\cite{KunduGGS22,BoseAKS20,inproceedings}).

At any moment in time, a robot can either be in idle state or in non-idle state. To move from the idle state to a non-idle state a robot first gets activated. The non-idle state of a robot is divided into three phases namely \textsc{Look} phase, \textsc{Compute} phase, and \textsc{Move} phase in this order. After activation, a robot first executes \textsc{Look} phase where it takes a snapshot of its surroundings to get the positions of other visible robots according to its own local coordinate system. Then it executes \textsc{Compute} phase, where it runs a distributed algorithm with the information from \textsc{Look} phase as input. Then as the output of the algorithm, the robot gets a position. After this, it executes the \textsc{Move} phase, where the robot moves to the position of the output. If the output position is the same as the current position, then it does not move.After executing the \textsc{Move} phase, the robot returns to the idle state until it activates again and performs the \textsc{Look}-\textsc{Compute}-\textsc{Move} phases. This is known as the LCM cycle. 

Now since one execution of the algorithm depends on the position of other robots, activation of the robots plays a huge role in designing swarm robot algorithms. The activation of a robot is controlled by an entity called a scheduler.  Mainly there are three types of scheduler models used massively in the literature. These three types of schedulers are,
\begin{itemize}
    \item \textit{Fully Synchronous Scheduler}($\mathcal{FSYNC}$): In fully synchronous scheduler, the time is divided into rounds of equal lengths. Each round is subdivided into three intervals for \textsc{Look}, \textsc{Compute}, and \textsc{Move} phase. These subdivisions are equal for each robot. Also, all robots are activated by an $\mathcal{FSYNC}$ scheduler at the beginning of each round.
    \item \textit{Semi-Synchronous Scheduler} ($\mathcal{SSYNC}$): Semi synchronous scheduler is somewhat similar to the fully synchronous scheduler. The only difference here is that at the beginning of each round, an $\mathcal{SSYNC}$ scheduler activates a  subset of robots. Note that if the subset is equal to the whole set of robots for each round, then we get the $\mathcal{FSYNC}$ scheduler. So $\mathcal{SSYNC}$ is more general than an $\mathcal{FSYNC}$ scheduler.
    \item \textit{Asynchronous scheduler} ($\mathcal{ASYNC}$): In the asynchronous scheduler, there is no notion of rounds. In a particular moment, a robot can be either idle or can be executing any one of \textsc{Look}, \textsc{Compute} or, \textsc{Move} phase. The duration of any of these phases is finite but unpredictable. This is the most general scheduler model which is also very feasible in terms of practical implementation. 
\end{itemize}

In this work, we are interested in the problem of pattern formation. Now there are two types of pattern formation problems on which there is a vast literature. The two types of pattern formation problems are as follows:
\begin{itemize}
    \item \textit{Arbitrary Pattern Formation} ($\mathcal{APF}$) where a set of robots, deployed in an environment, are provided with target pattern coordinates with respect to some global coordinate system. The robots do not have agreement to any global coordinate system however they each have their own local coordinate system. The aim is to provide a distributed algorithm for the robots so that first they can agree on a global coordinate and according to that coordinate system embed the target pattern and move to the target positions to form the pattern.(\cite{KunduGGS22,BoseAKS20,FlocchiniPSW08})
    \item \textit{Geometric Shape Formation }(e.g., Line Formation, Circle Formation, Uniform Circle Formation etc.). In this problem, the robots only know which shape to form but they are not provided with specific coordinates of the target according to some global coordinate system. For example, in circle formation problems the robots know that they have to form a circle but they have no agreement initially about the radius and center of the circle. These are decided autonomously by executing the algorithm (\cite{Sugihara1990,DefagoS08,AdhikaryKS21,inproceedings}).
\end{itemize}  
In this work, the main focus is this \textit{Geometric Shape Formation} problem on a discrete domain. To be more specific, this work investigates the problem of \textit{circle formation on an infinite grid by opaque fat robots} ($CF\_FAT\_GRID$) under an asynchronous scheduler. In the following subsection, we aim to offer a brief overview of how research on this issue has progressed, tracing its evolution from its inception to the contemporary landscape.
\subsection{Related Works}

Circle Formation is a well-known problem in distributed computing. It has been examined within both continuous and discrete contexts. The origin of the circle formation problem dates back to the work of Sugihara and Suzuki \cite{Sugihara1990}. While they proposed a heuristic algorithm, it produced an approximate circle. Subsequently, another more accurate approximation algorithm was introduced by Tanaka et al. \cite{tanaka}. 

The problem, of strategically positioning robots equidistantly along the circle's circumference is known as  Uniform Circle Formation. This notion was first explored by Suzuki and Yamashita \cite{suzuki99}. Later, Defago and Konogaya \cite{defago02} developed a circle formation algorithm that did not require robot orientation.

Defago and Konogaya's later work \cite{DefagoS08} presented a deterministic algorithm for non-uniform circle formation, where robots converge to an evenly spaced boundary configuration. Flocchini et al. \cite{FlocchiniPSV17} contributed an alternative uniform circle formation algorithm, removing certain assumptions, albeit restricted to cases where $n \neq 4$.

It is important to note that these aforementioned studies focused on the point robot model, with transparent robots operating in an Euclidean plane.

Feletti et al. \cite{inbook} first considered opaque robot models for solving uniform circle formation. In their work, opaque luminous point robots are initially placed on a plane, equipped with a light possessing 5 distinct colors under the $\mathcal{FSYNC}$ scheduler. Building on their prior research, their recent work \cite{app13137991} extended the solutions to incorporate asynchrony in the scheduler, using 19 colors for luminous, opaque point robots on a plane.

The utilization of opaque fat robots for circle formation was pioneered by the work in \cite{inproceedings}. This algorithm tackled circle formation under a limited visibility model while considering global coordinate agreement.

Within discrete domain, the circle formation problem emerged in \cite{Adhikary21}. Here, the authors considered luminous opaque point robots on an infinite grid, functioning under an asynchronous scheduler. An algorithm was introduced by them where the robots are equipped with a light possessing seven colors and have one axis agreement. In \cite{ito2022brief}, Ito et al. improved upon this concept by enhancing the existing algorithm for circle formation on an infinite grid. They achieved uniform circle formation with a diameter of $O(n)$ using five colors, and an alternative algorithm formed a uniform circle with a diameter of $O(n^2)$ with complete visibility, using just four colors.  In a recent work \cite{KunduGGS22}, fat roots are introduced in a discrete domain (infinite grid) for solving arbitrary pattern formation problem

To our current knowledge, the application of fat robots in a discrete domain to solve the circle formation problem remains unexplored. Thus, our research focuses on investigating this challenge by employing opaque fat robots on an infinite grid. The subsequent subsection provides a concise problem description and outlines the contributions of this work.
\subsection{Problem Description and Our Contribution}
This work considers the problem of circle formation on an infinite grid by luminous opaque fat robots under an asynchronous scheduler using only five colors. In a discrete environment, robots may not be able to form an exact circle. Thus, an approximated circle is considered (Definition~\ref{def:approx circle}) to translate the circle formation problem from the Euclidean plane into the discrete domain (here infinite grid). 

The robots are considered to be opaque disks of radius $rad \le \frac{1}{2}$. The center of the disk is considered to be the position of the robot. The robots operate in a \textsc{Look-Compute-Move} cycle under an asynchronous scheduler. Each robot has a light possessing five distinct colors. The robots do not agree on any global coordinate system; however, they agree on a common $x-$axis (i.e., a common left-right agreement).  The problem requires the robots to agree on a circle $\mathcal{CIR}$ and then move to the circumference of $\mathcal{CIR}$ and terminate.

The main challenge here is the visibility. Since the robots are opaque, a robot can obstruct the visibility of another robot. As a result, a robot may not have the total information of the current situation. This is also the case for point opaque robots. But the visual hindrance becomes more severe if the robots have dimension. For example, considering the opaque point robot model, a robot $r$ can not see another $r'$ iff there is another robot $r_o$ such that $r,r'$ and $r_o$ is colinear. But if we consider the opaque fat robots, a robot may not see another robot even when no three robots are colinear (Figure~\ref{fig:fat obstruction}). In this work, the provided algorithm carefully handles this issue.

\begin{figure}[h]
    \centering
    \includegraphics[width=3cm]{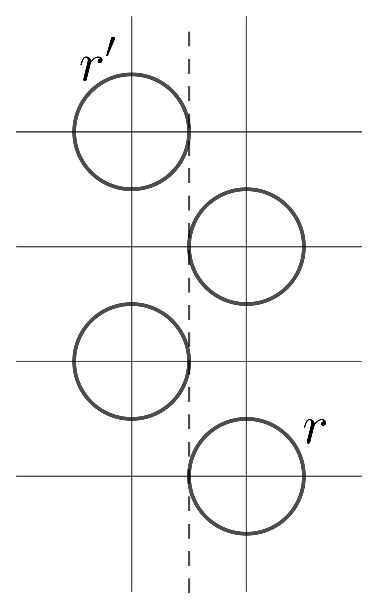}
    \caption{$r$ can not see $r'$ even though no three robots are colinear.}
    \label{fig:fat obstruction}
\end{figure}

Another challenge here is avoiding collisions. In the Euclidean plane, there are infinitely many paths between any two points, whereas, in a discrete domain, there are only finitely many paths between any two vertices. Thus, avoiding collision is easier in the Euclidean plane compared to the discrete domain. This challenge has been dealt with through carefully sequenced movements of robots. To simulate sequential movements in an asynchronous environment, we have used luminous robots.

The algorithm provided in this work is called the $CIRCLE\_FG$ algorithm. The $CIRCLE\_FG$ algorithm solves the $CF\_FAT\_GRID$ problem, considering a swarm of luminous, opaque fat robots with one-axis agreement and five distinct colors, from any initial configuration within finite time. A comparison table is provided so that readers can compare this work with some previously similar works.

\begin{table}[h]
  \centering
  \footnotesize
  \begin{tabular}{@{}cccccccc@{}}
    \toprule
    Paper& Domain& Scheduler& Axis Agreement& Vision& Dimension& Circle Type& \# colors \\
    \midrule
    \cite{app13137991} & Plane & $\mathcal{FSYNC}$ & No & Opaque & Point & Uniform & 5 \\
    \cite{app13137991} & Plane & $\mathcal{ASYNC}$ & No & Opaque & Point & Uniform & 19 \\
    \cite{AdhikaryKS21} & Grid & $\mathcal{ASYNC}$ & One axis & Opaque & Point & Non-Uniform & 7 \\
    \cite{ito2022brief} & Grid & $\mathcal{ASYNC}$ & One axis & Opaque & Point & Uniform & 5 \\
    \textbf{This Paper} & Grid & $\mathcal{ASYNC}$ & One axis & Opaque & Fat & Non-Uniform & 5 \\
    \bottomrule
  \end{tabular}
  \label{tab:table}
\end{table}

\section{Model and Definitions}
In this section, we first describe the models considered. Then, some definitions and notations are provided that have been used throughout the entire paper.
\subsection{Model}
\textbf{Infinite Grid:} An infinite grid is an infinite geometric graph $\mathcal{G} = (V,E)$ where vertices are points placed on $\mathbb{R}^2$ with coordinates $\{(a,b) : a \in \mathbb{Z}$ and $b \in \mathbb{Z}\}$. Also, two vertices are adjacent iff the Euclidean distance between them is one unit.

\textbf{Robot Model:} This work considers a set $\mathfrak{R}=\{r_1, r_2,\dots r_k\}$ of $k$ robots that are initially placed arbitrarily on $k$ vertices of an infinite grid $\mathcal{G}$. The robots are considered to be autonomous, anonymous, homogeneous and identical. The robots are not point. In this work, the robots are considered to be disks of radius at most $\frac{1}{2}$ unit. 

\textit{Axis and Unit Length Agreement:} Each robot has a local coordinate system where origin is the position of itself. There is no global agreement on the coordinates however the robots agree on the direction and orientation of the $x$-axis which is parallel to one of the grid lines and also on unit length. That implies the robots agree on left, right and on the distance between any two points, but do not agree on up and down.

\textit{Visibility Model:}
In this work, the robots have unlimited but obstructed visibility. A robot $r_i$ can see another robot $r_j$ if and only if there is a point on the perimeter of $r_i$, say $p_i$ and another point on the boundary of $r_j$, say $p_j$ such that the line segment $\overline{p_ip_j}$ does not intersect at a point on any other robot. From this visibility model, it follows that if $r_i$ can see $r_j$ then $r_j$ also sees $r_i$.

\textit{Memory and communication:} In this work, the robots are considered to be of the $\mathcal{LUMI}$ model. Thus the robots are able to remember some finite previous states and can communicate with visible robots using finite bit messages broadcasted using lights. Each of the robots are equipped with a light that can have 5 distinct colors from the set $ Col=\{$\texttt{off}, \texttt{chord}, \texttt{moving1}, \texttt{diameter}, \texttt{done}$\}$ one at a time. Upon activation from idle state, a robot executes according to the \textsc{Look}-\textsc{Compute}-\textsc{Move} cycle (LCM cycle) as described below.

 \textit{LCM cycle}: Upon activation  with a color $C_1 \in Col$, a robot $r$ first executes the \textsc{Look} phase. During this phase, the robot takes a snapshot of its surroundings and gets the positions of other visible robots according to its own coordinate system. Using this information as input, the robot then executes the  algorithm in the \textsc{Compute} phase. As an output of the algorithm, the robot gets a color $C_2 \in Col$ and a grid point at most one hop away from its current position. During the \textsc{Move} phase, the robot first changes its color to $C_2$ from $C_1$ (if different) and then moves to the new grid point (if different). After the \textsc{Move} phase is executed, the robot again returns to the idle state until it is activated again. A robot can move only along the edges of the grid, i.e., a robot can only move to one of the four adjacent vertices of its current position by moving once. The move is also considered to be rigid and instantaneous, i.e., a robot is always seen on the grid points.

\textbf{Scheduler Model:} The scheduler considered in this work is the most general asynchronous scheduler ($\mathcal{ASYNC}$). There is no agreement on rounds. The time taken by a robot to execute the \textsc{Look} phase, \textsc{Compute} phase, \textsc{Move} phase, and the time a robot remains idle is finite but unpredictable.

\subsection{Definitions and Notations}
\begin{definition}[\textbf{\textit{Grid Circumference}}]
\label{def:approx circle}
    Let $\mathcal{CIR}$ be a circle on the plane on which the infinite grid is embedded. Let $L_H$ be a horizontal grid line which intersects the circle $\mathcal{CIR}$ on at most two points $A$ and $A'$. Now,
    
    \begin{itemize}
        \item If $A$ and $A'$ are grid points, we say that only $A$ and $A'$ on $L_H$ are on the circumference of $\mathcal{CIR}$.
        \item  Otherwise, if $A$ and $A'$ are not grid points, we call a grid point $(a,b)$ on the circumference of $\mathcal{CIR}$ on $L_H$ if the line joining the grid points $(a,b)$, $(a+1,b)$ or $(a,b)$, $(a-1,b)$ contains exactly one of $A$ and $A'$.
    \end{itemize}
    The set of all such grid points on the circumference of the circle is called the Grid Circumference.
\end{definition}

For simplicity by the term ``on the circumference of the circle" we will always mean on some grid point which is in the set Grid Circumference.

Now, let us define the problem ($CF\_FAT\_GRID$) formallly .

\begin{definition}[\textbf{\textit{Problem Statement of $CF\_FAT\_GRID$}}]
Let a finite set of fat robots of same radius are initially located on distinct vertices of an infinite grid $\mathcal{G}.$ We say that the Circle formation by Fat robots on infinite grid ($CF\_FAT\_GRID$) is solved if there is an algorithm $\mathcal{A}$ such that after executing $\mathcal{A}$, for finite time the following conditions are satisfied.
\begin{enumerate}
    \item $ \mathcal{A}$ terminates. That is there exists a time $t$ such that at time $t$ all robots are terminated and does not execute $\mathcal{A}$ anymore.
    \item At time $t$, all robots are located on the grid circumference of the same circle, say $\mathcal{CIR}$.
\end{enumerate}
\end{definition}

\begin{definition}[\textbf{\textit{Configuration}}]
       Let $\mathcal{G}=(V,E)$ be an infinite grid. Let $f: V \rightarrow \{0,1\} \times (Col \cup \{NULL\}) $ be a function such that
 \begin{equation*}
 f(a,b) = \begin{cases}
       (0,NULL) & \text{if there is no robot on (a,b)} \\
       (1,c) & \text{if a robot is on $(a,b)$ with color $c \in Col$}
     \end{cases}
\end{equation*}
Then we call the pair $(\mathcal{G},f)$ a configuration which is denoted as $\mathcal{C}$. A configuration at time $t$ is denoted as $\mathcal{C}(t)$.
\end{definition}

We have used some notations throughout the paper. A list of these notations is mentioned in the following table.
\newpage
\begin{table}[h]
  \centering
  
  \begin{tabular}{@{}ll@{}}
    \toprule
    Notation & Description\\
    \midrule
    $\mathcal{L}_1$ & First vertical line on left that contains at least one robot.\\
    $\mathcal{L}_i$ & $i$-th vertical line on the right starting from $\mathcal{L}_1$.\\
    $\mathcal{L}_V(r)$ & The vertical line on which the robot $r$ is located.\\

    $\mathcal{L}_H(r)$ & The horizontal line on which the robot $r$ is located.\\

     $\mathcal{L}_I(r)$ & The left immediate vertical line of robot $r$ which has at least one robot on it. \\

     $\mathcal{R}_I(r)$&  The right immediate vertical line of robot $r$ which has at least one robot on it. \\

    ${H}_L^O(r)$ &  Left open half for the robot $r$. \\

     ${H}_L^C(r)$ & Left closed half for the robot $r$ (i.e  ${H}_L^O(r) \cup \mathcal{L}_V(r)$).\\

     ${H}_B^O(r)$ &  Bottom open half for the robot $r$.\\

     ${H}_B^C(r)$ &  Bottom closed half for the robot $r$ (i.e  ${H}_B^O(r) \cup \mathcal{L}_H(r)$).\\

     ${H}_U^O(r)$ & Upper open half for the robot $r$.\\

    ${H}_U^C(r)$ & Upper closed half for the robot $r$ (i.e  ${H}_U^O(r) \cup \mathcal{L}_H(r)$). \\
    \bottomrule
  \end{tabular}
  \label{tab:table}
\end{table}

    



\begin{definition}[\textbf{\textit{Phase 1 Final Configuration (P1FC)}}]
We call a configuration a Phase 1 Final Configuration if the following conditions hold
\begin{enumerate}
    \item $\mathcal{L}_2$ has exactly two robots, say $r_1$ and $r_2$, with color \texttt{diameter}. All other robots are on either $\mathcal{L}_1$ or $\mathcal{L}_3$ with color \texttt{chord}.
    \item All robots on $\mathcal{L}_1$ and $\mathcal{L}_3$ are strictly between $\mathcal{L}_H(r_1)$ and $\mathcal{L}_H(r_2)$.
\end{enumerate}
    
\end{definition}

\begin{definition}[]
    Let $t_b$ be the time when a robot moves first from the initial configuration. We say that $\mathcal{L}_1$ is \textit{fixed} at a time $t_f > t_b$ if both of the following conditions hold:
    \begin{enumerate}
        \item $\forall t \in $ $[t_b, t_f)$, $\mathcal{C}(t)$ has a robot that is not on $\mathcal{L}_1$.
        \item From time $t_f$ onwards, no robot from $\mathcal{L}_1$ moves left until all robots move to $\mathcal{L}_1$.
    \end{enumerate}

\end{definition}

\begin{definition}[\textbf{\textit{Terminal Robot}}]
   In a configuration $\mathcal{C}$, a robot $r$ is called a terminal robot on $\mathcal{L}_V(r)$ if there is no robot either above or below $r$ on $\mathcal{L}_V(r)$.
\end{definition}
\section{$CIRCLE\_FG$ Algorithm}
 
In this section, we propose an algorithm called $CIRCLE\_FG$ that solves the $CF\_FAT\_GRID$ problem in finite time.
First, we discuss the outline of the algorithm and then discuss it in detail later in this section.

\textbf{Outline of the algorithm:} The algorithm works in two phases. In Phase 1, the robots first form a vertical line, say $L$. The two extreme robots on the line change their color to \texttt{diameter}. The  other robots on line $L$ move according to the function \textsc{ChordMove} to the vertical lines that are one hop away from $L$ thus forming a P1FC (Figure~\ref{fig:Whole Phase 1}). Note that all robots that now see both robots of color \texttt{diameter} agree on the diameter of the circle, and consequently, on the circle itself. So in Phase 1, the diameter is formed on which all robots can agree. Then in Phase 2, the robots with color \texttt{chord} first move to the vertical lines that are $\frac{d}{2}$ distance away from the agreed diameter on both sides of it ($d$ is the length of the diameter) and change color to \texttt{off}. "From this configuration, all robots change their color to \texttt{moving1} and move towards the diameter by forming a line on each vertical grid line in between. Let $L_m$ be a vertical line where all robots of color \texttt{moving1} to the left (or right) of the agreed diameter have formed a line. Starting from $L_m$, all robots that are strictly inside the agreed circle move further from the diameter after changing their color to \texttt{done} in order to reach their corresponding positions on the circle. Following this, the remaining robots on $L_m$ move towards the agreed diameter and establish a new line on the next vertical grid line towards the diameter. In this manner, all robots move to their designated positions on the circle and terminate their movement. This entire process is illustrated in Figure~\ref{fig:whole phase 2 }. This is a very brief description of the algorithm. In the following two subsections, we will provide detailed descriptions of Phase 1 and Phase 2, along with explanations of the correctness of each phase. 
\begin{figure}[h]
    \centering
    \includegraphics[width=8cm]{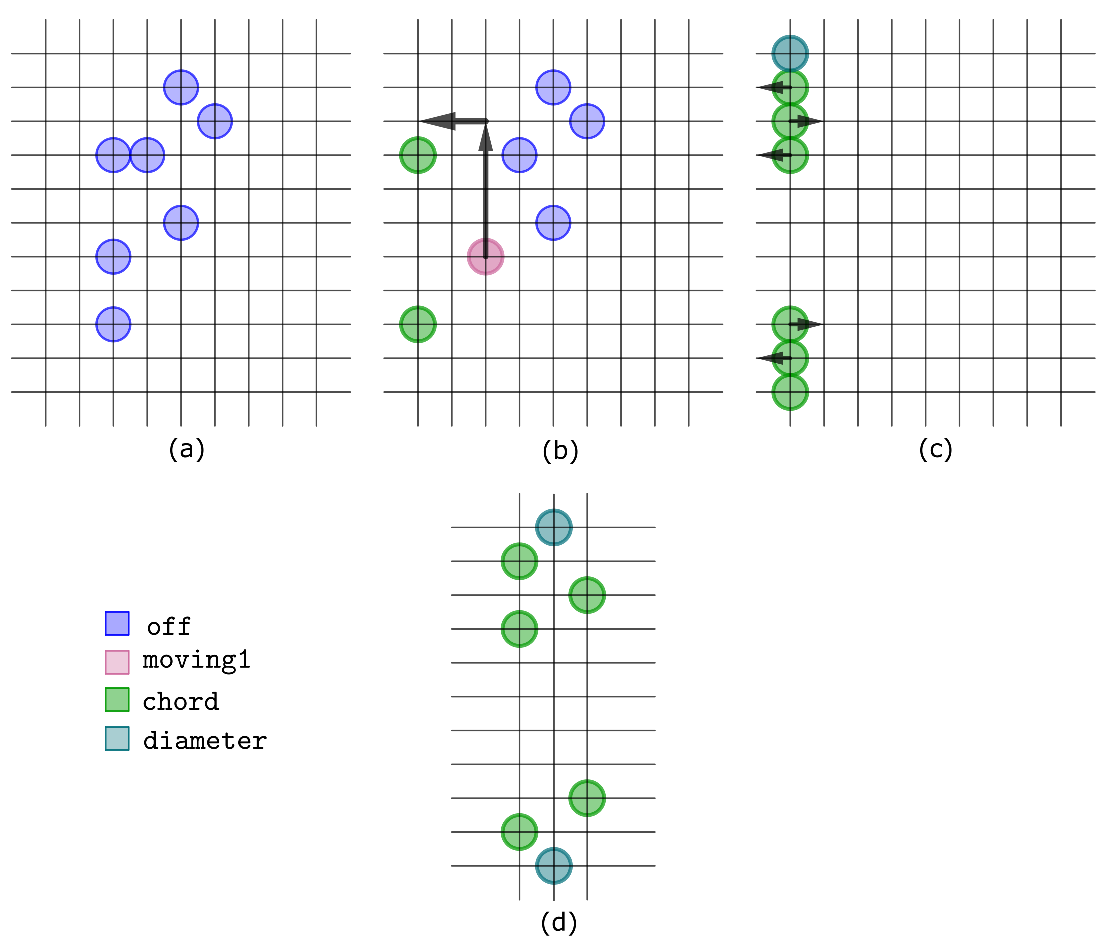}
    \caption{\textbf{Outline of Phase 1:} \textbf{(a)} Initial configuration $\mathcal{C}(0)$. \textbf{(b)} Two terminal robots of $\mathcal{L}_1$ in $\mathcal{C}(0)$ moved two hop and changes color to \texttt{chord}. A terminal robot of color \texttt{off} on nearest vertical line of $\mathcal{L}_1$ changes color to \texttt{moving1} and moves to $\mathcal{L}_1$. \textbf{(c)} All robots of color \texttt{off} moves to $\mathcal{L}_1$ in a similar way and forms a single line. The terminal robots on this line changes color to \texttt{diameter} eventually.(here in this specific example, due to asynchrony only one terminal robot changes color to \texttt{diameter} first). Next all non terminal robots move either left or right. (\textbf{d}) P1FC is formed.}
    \label{fig:Whole Phase 1}
\end{figure}
\begin{figure}[h]
    \centering
    \includegraphics[width=15.5cm]{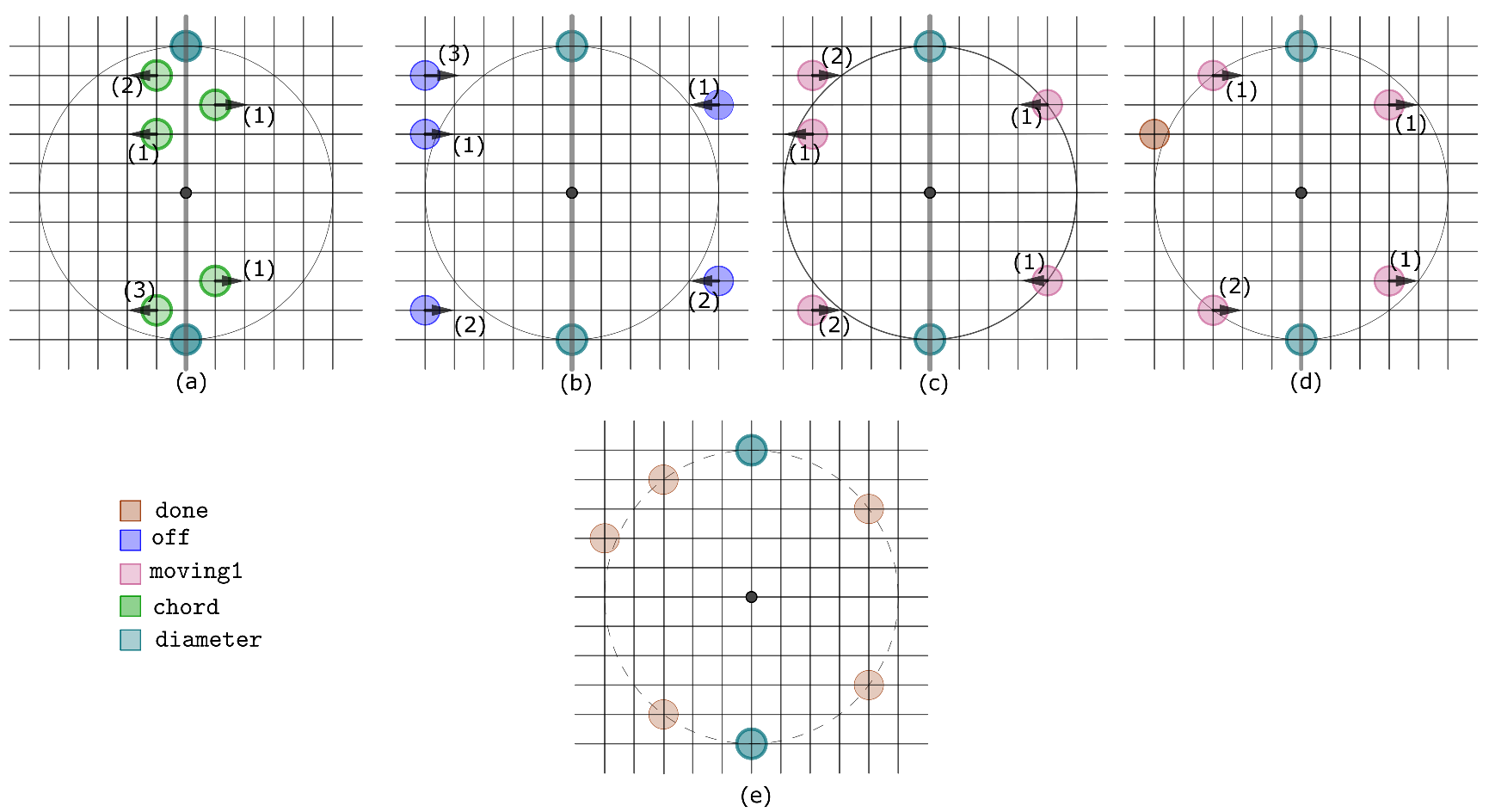}
    \caption{\textbf{Outline of Phase 2:} \textbf{(a)} From P1FC, all robots with color \texttt{chord} from one vertical line moves to another vertical line further from the diameter until it reaches $\frac{d}{2}$ distance from diameter where $d$ is the diameter length. The number $(i)$ associated with a robot denotes that it moved after $i-1$ robots has already moved. \textbf{(b)} Before reaching the vertical line at distance $\frac{d}{2}$, all robots change their color to \texttt{off}. Now all robots of color \texttt{off} moves one hop to the vertical line towards diameter after changing the color to \texttt{moving1}. \textbf{(c)-(d)} When all robots of color \texttt{moving1} are on a line the robots that are strictly inside the circle moves away from diameter after changing the color to \texttt{done} and the other robots of color \texttt{moving1} moves to the next vertical  line towards diameter after all inside robots moved. \textbf{(e)} Circle is formed. }
    \label{fig:whole phase 2 }
\end{figure}
\label{lf}

\subsection{Phase 1}
\subsubsection{Brief description of Phase 1}

In Phase 1, initially all the robots have color \texttt{off} and are placed arbitrarily on the grid. There can be at most two robots that are terminal on $\mathcal{L}_1$ of the initial configuration. Upon activation these robots will see their left open half is empty and one of their upper or bottom halves of their corresponding vertical line is empty. For this view the robot change their color to \texttt{moving1}
 from \texttt{off} and moves left shifting the line $\mathcal{L}_1$. Except the terminal robots on $\mathcal{L}_1$ of the initial configuration other robots of color \texttt{off} change their color to \texttt{moving1} only when they are terminal on their corresponding vertical line and sees all robots on their left immediate vertical line has color \texttt{chord}. So, unless the terminal robots on $\mathcal{L}_1$ of the initial configuration change their color to \texttt{chord} no other robots do anything even if they are activated.

 Let $r_1$ and $r_2$ be two terminal robots on $\mathcal{L}_1$ in the initial configuration. If $\mathcal{L}_1$ of the initial configuration contains more than two robots then there must be another robot except $r_1$ and $r_2$, say $r$, with color \texttt{off}. Now upon activation let $r_1$ has moved left once. Now when it is activated again, it will see $r$ on its $\mathcal{R}_I(r_1)$ having color \texttt{off} and also $\mathcal{R}_I(r_1)$ is one hop away from $\mathcal{L}_V(r_1)$. In this view if $H_L^O(r_1)$ is empty then, $r_1$ moves left again shifting the $\mathcal{L}_1$ further left. The target of this move is to make the distance between $\mathcal{L}_1$ and the first vertical line that contains a robot of color \texttt{off} more than one. Let when $r_1$ moved left further $r_2$ was on $\mathcal{L}_V(r)$ with a pending movement due to asynchrony. Then, it will move left and see all robots on $\mathcal{L}_I(r_2)$ has color \texttt{moving1} and it is singleton on $\mathcal{L}_V(r_2)$. For this view $r_2$ moves to left again and moves to $\mathcal{L}_1$ along with $r_1$. A robot with color \texttt{moving1} can change its color to \texttt{chord} for two possible views. For the first one it has to be on $\mathcal{L}_1$ and has to see another robot with color \texttt{chord} on $\mathcal{L}_1$. For the other one it has to see its right immediate vertical line, which is at least two hop away, has at least a robot of color \texttt{off}. So when $r_1$ moved left twice from its initial position upon its next activation activation, if $r_2$ is still in $\mathcal{L}_v(r)$ and has not yet changed its color to \texttt{moving1}, $r_1$ changes its color to \texttt{chord} otherwise $r_2$ reaches $\mathcal{L}_V(r_1)$ and eventually both of them change their color  to \texttt{chord} before any other robot does anything. Now a robot, say $r'$, which has changed its color to \texttt{moving1} from \texttt{off} on $\mathcal{L}_k$ by seeing all robots of color \texttt{chord} on $\mathcal{L}_I(r')$ must be terminal on $\mathcal{L}_k$ for some $k>1$. Now, if upon activation $r'$ sees all visible robots on $\mathcal{L}_I(r')$ has color \texttt{chord} and no robot of color \texttt{chord} on $\mathcal{L}_V(r')$ (This condition is to stop robots of color \texttt{moving1} to move further left from $\mathcal{L}_1$), then it can have either both of  $H_U^C(r') \cap \mathcal{L}_I(r')$ and $H_B^C(r') \cap \mathcal{L}_I(r)$ non empty or, empty. For this case, if there is another robot on $\mathcal{L}_V(r')$, then $r'$ moves along $\mathcal{L}_V(r')$ opposite to that robot, otherwise it moves along its positive Y-axis. After finite moves one of $H_U^C(r') \cap \mathcal{L}_I(r')$ and $H_B^C(r') \cap \mathcal{L}_I(r)$ must become empty for $r'$ when it moves left. Note that all robots of color \texttt{off} decides that it is in Phase 1 when they see no robots of color \texttt{diameter}. A robot with color \texttt{moving1} decides it is in Phase 1 when it sees no robot of color \texttt{diameter} or when it sees a robot of color \texttt{diameter} on its own vertical line. If a robot of color \texttt{moving1} sees  a robot with color \texttt{diameter} on its own vertical line, it changes its color to \texttt{diameter}.

Now a robot with color \texttt{chord} changes its color to \texttt{diameter} when it sees at least one robot of color \texttt{chord} on its corresponding vertical line and sees no other robot on its left and right open halves and one of upper or bottom closed halves. A robot, say $r_c$, with color \texttt{chord} executes \textsc{ChordMove} when it is not terminal on $\mathcal{L}_V(r_c)$ and sees at least one robot of color \texttt{diameter} on  $\mathcal{L}_V(r_c)$. On the otherhand if $r_c$ is terminal on $\mathcal{L}_V(r_c)$ and sees a robot of color \texttt{diameter} on $\mathcal{L}_(r_c)$ then it changes its color to \texttt{diameter} from \texttt{chord}.
a robot with color \texttt{chord} distinguishes Phase 1 when it does not see any robot of color \texttt{diameter} and also sees at least one robot with \texttt{diameter} color on its own vertical line.
We now describe the \textsc{ChordMove} subroutine.

\textit{\textsc{ChordMove} Subroutine:} A robot $r$ with color \texttt{chord} executes the subroutine \textsc{ChordMove} when it sees at least one robot with color \texttt{diameter} and is not terminal on $\mathcal{L}_V(r)$. while executing \textsc{ChordMove}, if a robot, say $r$, sees only one robot $r_1$ with color \texttt{diameter} then it checks if there is any other robot between the horizontal lines passing through $r_1$ and $r$ i.e., $\mathcal{L}_H(r_1)$ and $\mathcal{L}_H(r)$ respectively. If there is no robots in the above mentioned region then $r$ moves left. Now if there are robots between the mentioned region, then the following procedure takes place. If $r_N$ be the nearest of $r$ which is in the mentioned region and if $r_N$ is on $H_L^O(r)$ then $r$ moves right otherwise $r$ moves left. Now, if a robot $r$ sees both the robots, say $r_1$ and $r_2$, with color \texttt{diameter} then, $r$ finds its direction to move as stated above considering both the regions between $\mathcal{L}_H(r_1)$, $\mathcal{L}_H(r)$    and $\mathcal{L}_H(r_2)$, $\mathcal{L}_H(r)$ . If for both the direction considering both the regions are same then $r$ moves according to that direction otherwise it moves left. Note that after this procedure is complete, $\mathcal{L}_V(r_1)$ has only two robots $r_1$ and $r_2$ both having the color \texttt{diameter} and $\mathcal{L}_I(r_1)$ and $\mathcal{R}_I(r_1)$ contains all the robots with color \texttt{chord}. Also observe that, difference between number of robots on $\mathcal{L}_I(r_1)$ and $\mathcal{R}_I(r_1) \le 2$. Note that, a robot, say $r$, with color \texttt{chord} does nothing after it has already moved once executing \textsc{ChordMove} until the whole procedure is complete. This is because after $r$ has moved once, it sees $r_1$ or $r_2$ not on $\mathcal{L}_V(r)$ and $\mathcal{L}_V(r_1)$ has other robots except $r_1$ and $r_2$ until all non terminal robots on $\mathcal{L}_V(r_1)$ executes this subroutine exactly once.

Robots begin to execute Phase 2 from a Phase 1 Final Configuration. So, we have to ensure that by executing Phase 1 robots will eventually form a Phase 1 Final Configuration. This is ensured in the subsequent correctness section. The pseudo code of Phase 1 is presented in Algorithm~\ref{Algo_Phase1}.

\begin{figure}[]
  \centering
   \begin{minipage}{1\linewidth}
  \begin{algorithm}[H]
  \small
    \SetKwInOut{Input}{Input}
    \SetKwInOut{Output}{Output}
    \SetKwProg{Fn}{Function}{}{}
    \SetKwProg{Pr}{Procedure}{}{}

    \Pr{\textsc{Phase1()}}{

    $r \leftarrow$ myself
  
  \uIf{$r.color = $ \texttt{off}}
  {
    \If{$r$ sees no robot with color \texttt{diameter}}
    {
        \If{There is no robot i $H_L^O(r)$ and $r$ is terminal}
        {
            $r.color \leftarrow $ \texttt{moving1}\;
            move left\;
        }
        \ElseIf{$\mathcal{L}_I(r)$ is at least two hop away and all visible robots and all visible robots on $\mathcal{L}_I(r)$ has color \texttt{chord}}
        {
            \If{$r$ is terminal}
            {
                $r.color \leftarrow $ \texttt{moving1}\;
            }
        }
    }
  }
  \uElseIf{$r.color =$ \texttt{moving1}}
  {
    \eIf{$r$ sees no robot with color \texttt{diameter} on $\mathcal{L}_V(r)$}
    {
        \uIf{ all visible robots on $\mathcal{L}_I(r)$ has color \texttt{chord} and $r$ sees no robot of color \texttt{chord} on $\mathcal{L}_V(r)$}
        {
            \eIf{$H_U^C(r)\cap \mathcal{L}_I(r)$ and $H_B^C(r)\cap \mathcal{L}_I(r)$ both are non empty }
            {
                \eIf{There is a robot $r'$ on $\mathcal{L}_V(r)$}
                {
                    move opposite to $r'\;$
                }
                {
                    move according to positive $Y-$ axis\;
                }
            }
            {
                move left\;
            }
        }
        \uElseIf{ ($r$ is singleton on $\mathcal{L}_V(r)$ and all visible robots on $\mathcal{L}_I(r)$ has color \texttt{moving1}) or, (distance of $\mathcal{R}_I(r)$ having a robot with color \texttt{off} = 1 and $H_L^O(r)$ is empty.}
        {
            move left\;
        }
        \ElseIf{$H_L^O(r)$ is empty}
        {
            \uIf{sees a robot with color \texttt{chord} on $\mathcal{L}_V(r)$}
            {
                $r.color \leftarrow $ \texttt{cord}\;
            }
            \ElseIf{distance of $\mathcal{R}_I(r)$ having a robot of color \texttt{off} $\ge 2$}
            {
                $r.color \leftarrow$ \texttt{chord}\;
            }
        }
    }
    {
        $r.color \leftarrow$ \texttt{diameter}\;
        terminate\;
    }
  }
  \ElseIf{$r.color = $ \texttt{chord}}
  {
    \uIf{$r$ sees no robot with color \texttt{diameter}}
    {
        \If{There is a robot with color \texttt{chord} on $\mathcal{L}_V(r)$, there is no robot on $H_L^O(r), H_R^O(r)$ and $H(r)$ where $H(r) \in \{H_B^C(r), H_U^C(r)\}$}
        {
             $r.color \leftarrow$ \texttt{diameter}\;
             terminate\;
        }
    }
    \ElseIf{$r$ sees a robot with color \texttt{diameter} on $\mathcal{L}_V(r)$}
    {
        \eIf{$r$ is terminal}
        {
            $r.color \leftarrow$ \texttt{diameter}\;
            terminate\;
        }
        {
            Execute \textsc{ChordMove}()\;
        }
    }
  }

  }

    \caption{\textbf{Phase 1}}
    \label{Algo_Phase1}
\end{algorithm}

 \end{minipage}
\end{figure}
\subsubsection{Correctness of Phase 1}
Phase 1 is divided into two parts. In the first part, the robots first form a line where all robots have either the color \texttt{chord} or \texttt{moving1}. To be specific only terminal robots can have the color \texttt{moving1} on that line. Now at least one of the terminal robots eventually change its  color to \texttt{diameter}. Then in the second part, the non terminal robots with color \texttt{chord} moves left or right once, reaching either $\mathcal{L}_1$ or $\mathcal{L}_3$ and thus forming a P1FC eventually. 

So, first we have to show that all robots must move to a single line eventually where all robots have color either \texttt{chord} or \texttt{moving1}. We prove this by ensuring that all robots with color \texttt{off} eventually change its color to \texttt{moving1} (Lemma~\ref{lemma: off changes to moving11=}). Then all robot with color \texttt{moving1} moves to $\mathcal{L}_1$ (Lemma~\ref{lemma:moving1 on L_j moves to L_(j-1)}). For this, we also ensured that $\mathcal{L}_1$ becomes fixed after a finite time (Lemma~\ref{lemma: L1 is fixed}), otherwise a potential livelock situation may occur.

In the following, we have stated some observations proved some claims  which will be needed to prove the above mentioned lemmas.

\begin{observation}
\label{observation:2 hop Visible}
    A robot $r$ can see all robots of $\mathcal{L}_I(r)$ (resp. $\mathcal{R}_I(r)$) if $\mathcal{L}_I(r)$ (resp. $\mathcal{R}_I(r)$) is at least two hop away from $r$.
\end{observation}
\begin{observation}
    \label{observation: right half of off is  all off}
    If $r$ be a robot of color \texttt{off} executing Phase 1 such that $H_R^O(r)$ is non empty, then all robots on $H_R^O(r)$ must have color \texttt{off} 
\end{observation}

\begin{observation}
    \label{observation: before chord atmost two robot on L1} After a move by any robot from the initial configuration and before any robot changes its color to \texttt{chord}, $\mathcal{L}_1$ can have at most two robots of color \texttt{moving1} and all other robots has color \texttt{off}.
\end{observation}

\begin{claim}
     \label{Lemma: r sees off on right imidiate}
    Let $r$ be a robot with color \texttt{moving1} such that $\mathcal{R}_I(r)$ is one hop away from $\mathcal{L}_V(r)$ and there is at least one robot with color \texttt{off} on $\mathcal{R}_I(r)$. If $r$ is activated in this configuration, then $r$ always sees a robot having color \texttt{off} on $\mathcal{R}_I(r)$ in Phase 1.
\end{claim}

\begin{proof}
    Let $r$ be a robot with color \texttt{moving1}. Let $\mathcal{R}_I(r)$ be one hop away from $\mathcal{L}_V(r)$ which has a robot, say $r'$ of color \texttt{off} on it. If possible $r$ does not see any robot with color \texttt{off} on $\mathcal{R}_I(r)$. That is $r$ does not see $r'$ upon activation, say at a time $t >0$.
    Let $\mathcal{L}_H(r)$ and $\mathcal{L}_H(r')$ be two horizontal lines passing through $r$ and $r'$ respectively. Then the above assumption is true only when $\mathcal{L}_V(r)$ and $\mathcal{L}_V(r')$ each contains at least one robot between $\mathcal{L}_H(r)$ and $\mathcal{L}_H(r')$. Let $r_1$ and $r_2$ be two such robots on $\mathcal{L}_V(r)$ and $\mathcal{L}_V(r')$ respectively. without loss of generality, let $r$ can see $r_2$ and $r'$ can see $r_1$. Note that $r_2$ must be of color \texttt{moving1} at time $t$. Also $r_1$ either have color \texttt{off} or color \texttt{moving1} at time $t$ ($r_1$ can not be of color \texttt{chord} at time $t$ as $r'$ is on $\mathcal{R}_I(r_1)$ having color \texttt{off} at time $t$).
    
    Now, if $r_1$ has color \texttt{off} then in the interval $(0,  t)$, $\mathcal{L}_V(r_1)$ has not changed. Now in this interval, all robots on $\mathcal{L}_V(r_1)$ must have color \texttt{moving1} or \texttt{off}. So in the interval, $(0, t)$, $r_2$ can never see all robots with color \texttt{chord} on $\mathcal{L}_I(r_2)$ and thus can not change its initial color \texttt{off} to \texttt{moving1} in the mentioned interval. So, $r_2$ can not be of color \texttt{moving1} at time $t$ which is a contradiction to our assumption.

    So, let $r_1$ has color \texttt{moving1} at time $t$. Now, since at time $t$, $r_2$ has color \texttt{moving1}, there is a time $t_1 <t$ when $r_2$ has color \texttt{off}, is terminal on $\mathcal{L}_V(r_2)$ and either sees  $H_L^O(r_2)$ is empty or sees all robots  on $\mathcal{L}_I(r_2)$ having color \texttt{chord}. This implies $r$ and $r_1$ must have moved to $\mathcal{L}_I(r_2)$ after time $t_1$. Thus, at time $t_1$ all of $r, r_1$ and $r_2$ were along with $r'$ on $\mathcal{L}_V(r_2) = \mathcal{L}_V(r')$. Now since, $r$ and $r_1$ moves to $\mathcal{L}_I(r_2)$, they must have changed their color to \texttt{moving1} from initial color \texttt{off}. There are three cases. Firstly, let both $r$ and $r_1$ get activated and see that they are terminal on $\mathcal{L}_V(r')$ and change their color to \texttt{moving1} on or after time $t_1$. This is not possible as in this case $r,r_1$ and $r_2$ all have to be terminal on $\mathcal{L}_V(r')$ at time $t_1$ which is not possible. Secondly, let both $r$ and $r_1$ be activated and see themselves terminal on $\mathcal{L}_V(r')$ and change their color to \texttt{moving1} before $t_1$. This case is also impossible as at time $t_1$ since $r$ and $r_1$ are still on $\mathcal{L}_V(r')$, $r_2$ can not see itself as a terminal robot which is a contradiction. So first, let us assume  $r$ changed its color before time $t_1$ and $r_1$ changed its color after time $t_1$ and before time $t$. We claim that, for $r_1$ to change its color after time $t_1$, $r$ must have to move left from $\mathcal{L}_V(r')$. 
     \begin{figure}[!ht]
        \centering
        \includegraphics[height=5cm, width =5cm]{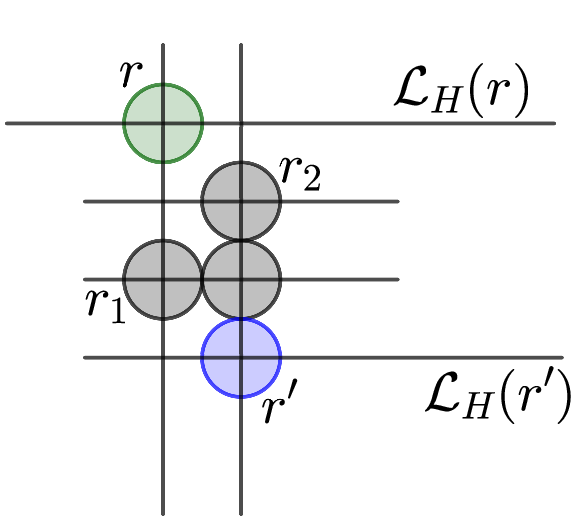}
        \caption{$r$ has color \texttt{moving1}, $r'$ has color \texttt{off} at time $t$ and $r$ can not see $r'$.}
        \label{fig:Moving1SeesOffOnRight}
    \end{figure}
    This is because, at time $t_1$, $r$ and $r_2$ both are terminal on $\mathcal{L}_V(r')$ so until $r$ moves $r_1$ can not become terminal and thus can not change its color. Now after time $t_1$ if $r_1$ is activated before time $t$, even if it is terminal now it will not change its color to \texttt{moving1} as $\mathcal{L}_I(r_1)$ is exactly one hop away (less than two hop away). Thus, even if $r_1$ is activated it does not change its color from \texttt{off} until time $t$. So at time $t$, $r$ can see $r_2$ with color \texttt{off} contrary to our assumption. Now, if $r_1$ had changed its color to \texttt{moving1} before time $t_1$ and $r$ after time $t_1$, then by similar argument it can be said that at time $t$, $r$ stays at $\mathcal{L}_V(r')$. So, in this case, if $\mathcal{R}_I(r)$ is non-empty, all robots on $\mathcal{R}_I(r)$ will have color \texttt{off} (Observation~\ref{observation: right half of off is  all off}). Thus, $r$ must see at least one robot having color \texttt{off} on $\mathcal{R}_I(r)$ even if it is one hop away contrary to our assumption. Thus, for a robot $r$ if $\mathcal{R}_I(r)$ has a robot with color \texttt{off} and $\mathcal{R}_I(r)$ is one hop away then $r$ always sees at least one robot of color \texttt{off} on $\mathcal{R}_I(r)$.

\end{proof}

\begin{claim}
    \label{chord only after L1 is fixed}
    Before $\mathcal{L}_1$ is fixed, no robot in the configuration can have color \texttt{chord}.
\end{claim}
\begin{proof}
    Let there is a time $t$ when $\mathcal{L}_1$ is not fixed but the configuration at time $t$ has a robot, say $r$ with color \texttt{chord}. without loss of generality let $r$ be the first robot that changes its color from \texttt{moving1} to \texttt{chord}. So, there must be a time $t_0 <t$ when $r$ is activated on $\mathcal{L}_1$ with color \texttt{moving1} and changes its color to \texttt{chord}. Note that at time $t_0$, $\mathcal{L}_1$ is not fixed. Also, since there are more than 2 robots in the system, there must be at least one robot with color \texttt{off} which is not on $\mathcal{L}_1$ for the whole duration $[t_b, t_0)$. Thus there must be a robot on $\mathcal{L}_1$ with color \texttt{moving1} which moves left after $t_0$. Let $r'$ be that robot. Note that $r'$ can not be $r$ as a robot with color \texttt{chord} does not move in Phase 1. Now, if $r'$ is activated after time $t_0$ for executing the LCM cycle where it moves left from $\mathcal{L}_1$ then, upon activation, it must have seen a robot $r_1$ with color \texttt{off} on $\mathcal{R}_I(r') = \mathcal{L}_2$ which is one hop away from $\mathcal{L}_1$. But this is not possible because if $r_1$ is on $\mathcal{L}_2$ after time $t_0$ it must have been there at time $t_0$ also. So at time $t_0$, $r$ does not change its color to \texttt{chord} upon activation contrary to our assumption. So, let $r'$ be activated at a time $t' < t_0$ for executing the LCM cycle where it moves left from $\mathcal{L}_1$ after time $t_0$. This is only possible if at time $t'$, $\mathcal{R}_I(r') = \mathcal{L}_2$ had a robot, say $r_1$,  with color \texttt{off}. Now since $r'$ moves left after time $t_0$, at time $t_0$ upon activation $r$ must have seen a robot with color \texttt{off} on $\mathcal{R}_I(r) =\mathcal{L}_2$ which is one hop away from $\mathcal{L}_2$. Thus $r$ doesn't change its color to \texttt{chord} upon activation at time $t_0$. This is also a contradiction. Thus,  before $\mathcal{L}_1$ is fixed, no robot changes their color to \texttt{chord}.
\end{proof}
\begin{claim}
    \label{lemma: L1 only moving1 and chord} In Phase 1, between the time of first move by any robot from the initial configuration and the time when all robots move to a single line for the first time,  a robot on $\mathcal{L}_1$ have color either \texttt{moving1} or \texttt{chord}.
\end{claim}
\begin{proof}
    Let $t_b>0$ be the time when the first move by a robot happened from the initial configuration. Also, let $t_f>t_b$ be the time, when all robots move to a single line after time $t_b$ for the first time ($t_f$ can be infinite if all robots never move to a single line). Now, $\mathcal{L}_1$ can not have a robot with color \texttt{diameter} in the time interval $(t_b,t_f)$ as, in this interval no robot sees both its left and right open halves empty. Also, since at time $t_b$ at least one leftmost terminal robot moves left after changing its color to \texttt{moving1}, $\mathcal{L}_1$ also shifts left at time $t_b$. Note that, after this move, $\mathcal{L}_1$ does not have any robot with color \texttt{off}. So in the interval $(t_b,t_f)$,  $\mathcal{L}_1$ can not have any robot with color \texttt{off} as robots with color \texttt{off} never moves left to reach $\mathcal{L}_1$ (algorithm~\ref{Algo_Phase1}) and no robot of different color change their color to \texttt{off} in Phase 1. So, within the time interval $(t_b,t_f)$, $\mathcal{L}_1$ can have robots of color either \texttt{moving1} or of color \texttt{chord}. 
\end{proof}

\begin{lemma}
    \label{lemma: L1 is fixed}
   $\mathcal{L}_1$ can not shift left infinitely often without all robots being on $\mathcal{L}_1$.
\end{lemma}
\begin{proof}
    Let $t_b$ be the time when the first robot moves from the initial configuration. After $t_b$, suppose $\mathcal{L}_1$ shifts left infinitely often, while there remains at least one robot not positioned on $\mathcal{L}_1$ after $t_b$. This implies there is a robot, say $r$, which moves left from $\mathcal{L}_1$ infinitely often. Notably, robot $r$ must have the color \texttt{moving1}, and it retains this color without change. Now $r$ can move left from $\mathcal{L}_1$ only if it sees a robot, say $r'$, of color \texttt{off} on $\mathcal{R}_I(r)$ which is one hop away from $\mathcal{L}_1$ (Figure~\ref{fig:L1 fixed}). Let $t_0 > t_b$ be a time when $r$ is activated on $\mathcal{L}_1$ and observes $r'$ on $\mathcal{R}_I(r)$, situated at a distance of one unit from $\mathcal{L}_1$. In this case, $r$ moves left and shifts $\mathcal{L}_1$ to left along with it. Note that after this move is completed, no robot on $\mathcal{L}_1$ will ever see another robot of color \texttt{off} on $\mathcal{L}_2$ (as robots with color \texttt{off} never move left in Phase 1). Consequently, for all subsequent times $t_1 > t_0$, if $r$ is reactivated at $t_1$ on $\mathcal{L}_1$ with the color \texttt{moving1}, it will not perceive any robot with the color \texttt{off} on $\mathcal{L}_2$. As a result, it will not move left, contradicting our initial assumption. Hence, the leftward shift of $\mathcal{L}_1$ cannot continue infinitely without all robots being positioned on the same line.
    \begin{figure}[!ht]
        \centering
        \includegraphics[height=4cm,width=7cm]{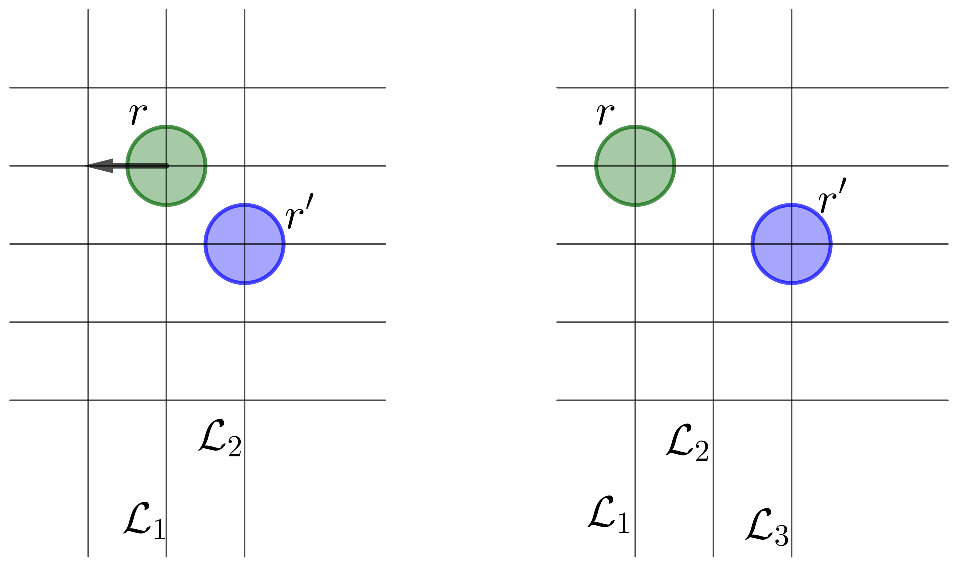}
        \caption{in the configuration on left, $r$ can see $r'$ with color \texttt{off} on $\mathcal{L}_2$ so it moves left. 
        in the configuration on right, $r$ can not see any robot on $\mathcal{L}_2$ with color \texttt{off} so it does not move}
        \label{fig:L1 fixed}
    \end{figure}
\end{proof}

 \begin{claim}
     \label{lemma: not more than two moving1 excluding L1}
     Excluding $\mathcal{L}_1$ there can not be more than two robots with color \texttt{moving1} in Phase 1.
 \end{claim}
 \begin{proof}
Suppose there are three robots, denoted as $r_1, r_2,$ and $r_3$, all with color \texttt{moving1}, positioned on $\mathcal{L}{k_1}, \mathcal{L}{k_2},$ and $\mathcal{L}{k_3}$ at time $t$, where $1 < k_1 \le k_2 \le k_3$. Initially, $r_1, r_2,$ and $r_3$ must all be located on the same vertical line. Otherwise, the rightmost robot among them in the initial configuration cannot transition its color from \texttt{off} to \texttt{moving1}, unless the leftmost robot(s) among the others reach $\mathcal{L}_1$ and alter their color to \texttt{chord}. However, this contradicts the observation that all three robots have the color \texttt{moving1} at time $t$. Therefore, let's assume that initially all of them are positioned on $\mathcal{L}_k$ with the color \texttt{off}. It is important to note that none of them moves left from $\mathcal{L}_k$ unless all of them are activated for their corresponding LCM cycle, in which they change their color to \texttt{moving1}. This implies that there exists a time $t' < t$ when all of $r_1, r_2,$ and $r_3$ are located on $\mathcal{L}_k$, and each of them either possesses the color \texttt{moving1} or becomes activated in their corresponding LCM cycle, during which they change their color to \texttt{moving1}. This implies $r_1, r_2$ and $r_3$ all are terminal on $\mathcal{L}_k$ at time $t'$ but this is not possible. Hence, excluding $\mathcal{L}_1$, there can not be more than two robots with color \texttt{moving1} in Phase 1.
 \end{proof}
\begin{claim}
    \label{lemma: moving1 terminal on L1 in phase 1}
    A robot with color \texttt{moving1} on $\mathcal{L}_1$ must be terminal on $\mathcal{L}_1$ in Phase 1.
\end{claim}
\begin{proof}
    From Observation~\ref{observation: before chord atmost two robot on L1} and Claim~\ref{chord only after L1 is fixed} it is evident that before $\mathcal{L}_1$ is fixed, it can have at most two robots and those are of color \texttt{moving1}. So before $\mathcal{L}_1$ is fixed, any robot with color \texttt{moving1} on $\mathcal{L}_1$ must be terminal on $\mathcal{L}_1$.

   Let $\mathcal{L}_1$ becomes fixed at a time $t_0$. Now let us assume $r$  be a robot with color \texttt{moving1} on $\mathcal{L}_1$ that is not terminal on $\mathcal{L}_1$ at a time $t > t_0$. This implies there must be another robot, say $r'$, which without loss of generality is directly below $r$ on $\mathcal{L}_1$ at the time $t$ and $r$ has color \texttt{moving1}. Now, $r$ and $r'$ must have moved to $\mathcal{L}_1$ from $\mathcal{L}_2$. Note that, in the initial configuration, $r$ and $r'$ can not be on two different vertical lines otherwise the rightmost robot among them can not change its color from \texttt{off} to \texttt{moving1} until the other one reaches $\mathcal{L}_1$ and change color to \texttt{chord}. 
   Also, even when they are on the same vertical line and one of them, say $r$, already moves left before the other one i.e., $r'$ wakes to change its color from \texttt{off} to \texttt{moving1}, it can't do so unless $r$ reaches $\mathcal{L}_1$ and change its color to \texttt{chord}. So, without loss of generality let us assume before one of $r$ and $r'$ moves from $\mathcal{L}_k$, the other robot must have been activated and seen $\mathcal{L}_1$ where all robots have color \texttt{chord}. This ensures that $r$ and $r'$ will change their color to \texttt{moving1} from \texttt{off}.
   Let in the initial configuration, $r$ and $r'$ were on the same line which is the line $\mathcal{L}_k$ at the time $t_0$ and $k \ge 3$ ($\mathcal{L}_2$ at time $t_0$ can not have any robot in the initial configuration). Also observe that, for $r$ and $r'$ to reach $\mathcal{L}_1$ they must move there from $\mathcal{L}_2$ (The lines $\mathcal{L}_i$ are denoted for the time when $\mathcal{L}_1$ becomes fixed). Since in the initial configuration $r$ and $r'$ were not on $\mathcal{L}_2$ of the current configuration, they must have color \texttt{moving1} while on $\mathcal{L}_2$. Suppose $r$ is located on $\mathcal{L}_2$ while $r'$ is on $\mathcal{L}_j$ with $j \ge 2$ and has color either \texttt{moving1} or is in the transitional LCM cycle where it would change its color to \texttt{moving1} on $\mathcal{L}_j$. Now for the former case, if $r$ is activated and it moves to $\mathcal{L}_1$  while $r'$ is idle then upon activation $r'$ can not move to $\mathcal{L}_{j-1}$ until $r$ changes its color to \texttt{chord}. For the latter case also, before the next activation of $r'$, if $r$ moves to $\mathcal{L}_1$, $r'$ can not move to $\mathcal{L}_{j-1}$ unless $r$ changes its color to \texttt{chord}. In either of these cases, when $r'$ reaches $\mathcal{L}_1$, $r$ must be of color \texttt{chord}, which contradicts our assumption.

    \begin{figure}[h]
       \centering
       \includegraphics[width=2cm]{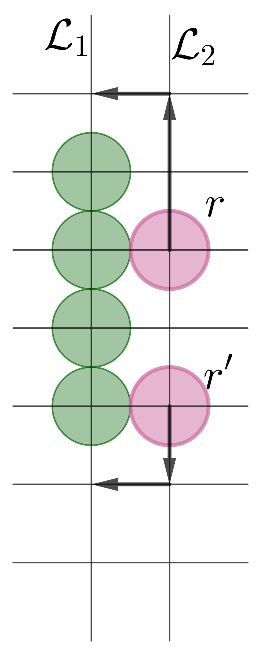}
       \caption{$r$ and $r'$ from $\mathcal{L}_2$ moves to $\mathcal{L}_1$ in such a way that at $\mathcal{L}_1$ all other robots are between $r$ and $r'$. Thus $r'$ can not be directly below or above $r$.}
       \label{fig:enter-label}
   \end{figure}
   Therefore, for both $r$ and $r'$ to reach $\mathcal{L}_1$ with the color \texttt{moving1}, there must exist a time $t_1 > t_0$ and $t_1 < t$ when both of them are on $\mathcal{L}_2$ with the color \texttt{moving1}. Furthermore, neither of them moves to $\mathcal{L}_1$ while the other one is inactive. Note that  $r'$ must have changed its color to \texttt{moving1} from \texttt{off} after seeing all robots with color \texttt{chord} as $\mathcal{L}_1$ of the initial configuration can not have more than two terminal robots of color \texttt{off}. This implies $\mathcal{L}_1$ must have at least one robot of color \texttt{chord} at the time when $r'$ moves on it after time $t_0$. Now since both $r$ and $r'$ are activated on $\mathcal{L}_2$ before any one of them moves left, they must see each other on $\mathcal{L}_2$ and move opposite of each other on $\mathcal{L}_2$ until one robot has moved above $\mathcal{L}_H(r_1)$ and the other moves below $\mathcal{L}_H(r_2)$, where $r_1$ is the uppermost robot on $\mathcal{L}_1$ and $r_2$ is the lowest robot on $\mathcal{L}_1$ before $r$ and $r'$ moves on to $\mathcal{L}_1$. Thus now when they move left on $\mathcal{L}_1$ they must have at least one robot with color \texttt{chord} between them. So at time $t$, $r'$  can not be directly below $r$. Thus A robot with color \texttt{moving1} on $\mathcal{L}_1$ must be terminal on $\mathcal{L}_1$.
\end{proof}

\begin{claim}
    \label{lemma: robot moving1 changes color to chord}  In Phase 1,  if the fixed $\mathcal{L}_1$ has a robot $r$ with color \texttt{moving1} which sees at least one robot that is not on $\mathcal{L}_1$ upon activation then $r$ changes its color to \texttt{chord} eventually.
 \end{claim}
    \begin{proof}
        If $\mathcal{L}_1$ has more than two robots then  $r$ must see a robot with color \texttt{chord} on $\mathcal{L}_V(r) = \mathcal{L}_1$ upon activation (Claim~\ref{lemma: moving1 terminal on L1 in phase 1} and Claim~\ref{lemma: L1 only moving1 and chord}) (Figure.~\ref{fig:moving1 On L1 to Chord}(a)). So, it must change its color to \texttt{chord}. Therefore, let $\mathcal{L}_1$ have at most two robots when $r$ is activated on the fixed $\mathcal{L}_1$ at a time, say $t_0$. If another robot has color \texttt{chord} on $\mathcal{L}_1$ when $r$ is activated at  $t_0$, by the same reason $r$ changes its color to \texttt{chord}. So let us assume that if there is another robot, say $r'$, on $\mathcal{L}_1$ at time $t_0$ then it has color \texttt{moving1} (Figure~\ref{fig:moving1 On L1 to Chord}(b)). In this configuration, there is no robot of color \texttt{chord}. So all other robots except $r$ and $r'$ are of color \texttt{off} in this configuration, thus they are at their initial positions. For this case, upon activation, $r$ must see at least one robot with color \texttt{off} on $\mathcal{R}_I(r)$ which is at least two hop away from $\mathcal{L}_1$ and thus $r$ changes its color to \texttt{chord}. Now at time $t_0$ upon activation if $r$ is singleton on $\mathcal{L}_1$ then there can be at most another robot with color \texttt{moving1} on $\mathcal{L}_j$ ($j > 1$) (Figure.~\ref{fig:moving1 On L1 to Chord}(c)). Now if there is no other robot with color \texttt{moving1} at time $t_0$ then $r$ sees $\mathcal{R}_I(r)$ has a robot of color \texttt{off} and it is at least 2 hop away from $\mathcal{L}_1$. Thus, in this case, $r$ changes its color to \texttt{chord}. Otherwise, the other robot, say $r'$ with color \texttt{moving1} is either singleton on some $\mathcal{L}_j$ or not ($j > 1$). Also in between $\mathcal{L}_1$ and $\mathcal{L}_j$, there is no other robot at time $t_0$. This is because, $\mathcal{L}_j$ at time $t_0$ must be either the vertical line $\mathcal{L}_1$ of the initial configuration or is strictly left of $\mathcal{L}_1$ of the initial configuration and since there is no robot with color \texttt{chord} on $\mathcal{L}_1$, except $r$ and $r'$ all robots has color \texttt{off} at time $t_0$ (this ensures all robots except $r$ and $r'$ never moved from their initial position until time $t_0$). If, at time $t_0$, $r'$ is singleton on some $\mathcal{L}_j$ ($j > 1$) then upon activation $r'$ moves left to $\mathcal{L}_{j-1}$. Note that if $r'$ is not on $\mathcal{L}_1$, $r$ on $\mathcal{L}_1$ does nothing even when activated. So $r'$ eventually reaches $\mathcal{L}_1$. Now when $r'$ reaches $\mathcal{L}_1$, by the above argument when $r$ activates next, it changes the color to \texttt{chord} as it sees $\mathcal{R}_I(r)$ has a robot with color \texttt{off} which is at least two hop away from $\mathcal{L}_1$. Now let at time $t_0$, $r'$ was not singleton on $\mathcal{L}_j$, then there must be another robot, say $r_1$, with color \texttt{off} on $\mathcal{L}_j$ which is seen by $r$ at time $t_0$. For this case also, $r$ changes its color to \texttt{chord}. 
    \end{proof}
\begin{figure}[h]
    \centering
    \includegraphics[width=10cm]{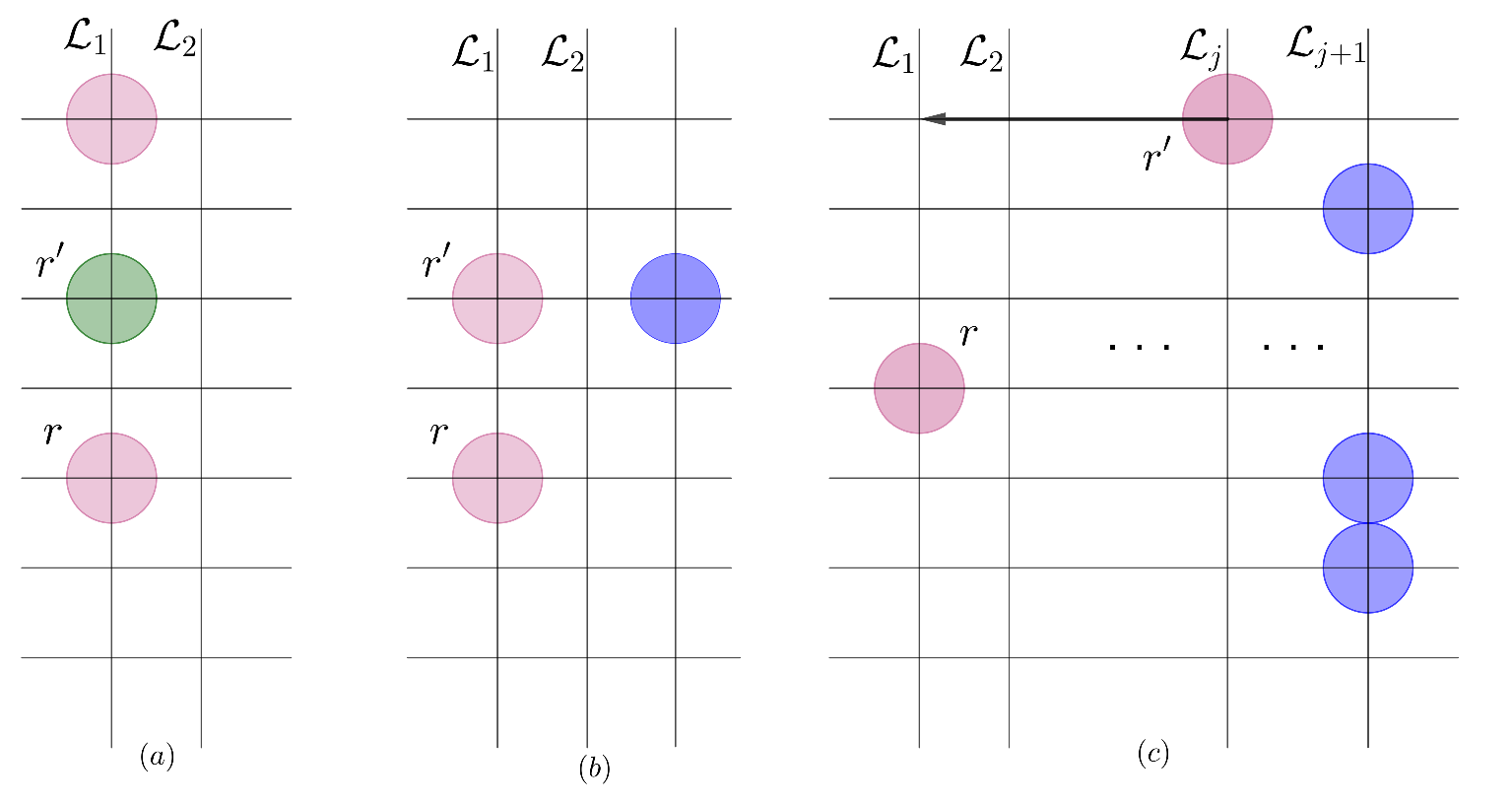}
    \caption{\textbf{(a)} $\mathcal{L}_1$ has more than two robots including $r$. Then $r$ sees $r'$ of color \texttt{chord} on $\mathcal{L}_1$. \textbf{(b)} There are exactly two robots, $r$ and $r'$ on $\mathcal{L}_1$ and both of them has color \texttt{moving1}.  \textbf{(c)} $r$ is singleton on $\mathcal{L}_1$ with color \texttt{moving1} and $r'$ is another robot with color \texttt{moving1} on $\mathcal{L}_j$. $r'$ moves to $\mathcal{L}_1$ and transforms into case (b). Here the blue color denotes color \texttt{off}} 
    \label{fig:moving1 On L1 to Chord}
\end{figure}
    \begin{lemma}
        \label{lemma:moving1 on L_j moves to L_(j-1)}
        In Phase 1, a robot with color \texttt{moving1} on $\mathcal{L}_j$ eventually moves to $\mathcal{L}_{j-1}$ where $j \ge 2$ if $\mathcal{L}_1$ is fixed.
    \end{lemma}
    \begin{proof}
        We will prove this using mathematical induction on $j$.
        
        \textbf{Base case:} In the base case we first establish that a robot $r$ with color \texttt{moving1} on $\mathcal{L}_2$ moves to $\mathcal{L}_1$ eventually. For that, let $r$ be a robot on $\mathcal{L}_2$ with color \texttt{moving1} at a time $t$. If possible, let $r$ never reaches $\mathcal{L}_1$. Note that there can be at most one robot, say $r'$, other than $r$ which is not on $\mathcal{L}_1$ and has color \texttt{moving1} at time $t$(Claim~\ref{lemma: not more than two moving1 excluding L1}). Now let at time $t$, $r'$ is on $\mathcal{L}_k$ for some $k \ge 2$. 

        \textit{Case 1:} Suppose $r'$ is on $\mathcal{L}_2$ along with $r$ (i.e., $k=2$) at time $t$. Note that $\mathcal{L}_2$ can not have any other robot of color \texttt{off} as it is strictly to the left of the  $\mathcal{L}_1$ in the initial configuration. So all robots on right of $\mathcal{L}_2$ must be of color \texttt{off} and does not do anything even if they are activated as whenever they are activated they never see any robot with color \texttt{chord} on their left immediate vertical lines. Now at time $t$, $\mathcal{L}_1$ either has all robot with color \texttt{chord} or,has all robots of color \texttt{chord} except at most two robots on the terminal. For the later case, by Claim~\ref{lemma: robot moving1 changes color to chord} all robot on $\mathcal{L}_1$ will have color \texttt{chord}  eventually at a time, say $t'$, where $t'> t$ (Figure~\ref{fig:moving1 move left}(a)). Until then $r$ and $r'$ does nothing even if they are activated. after $t'$ when $r'$ activates next it sees $\mathcal{L}_1$ has all robot with color \texttt{chord} and thus move left to $\mathcal{L}_1$. with similar argument after moving to $\mathcal{L}_1$ eventually $r'$ end up with color \texttt{chord}. next when $r$ is activated on $\mathcal{L}_2$ it must see that all robot on $\mathcal{L}_1$ has color \texttt{chord} and thus moves left contrary to the assumption.
        
        \textit{Case 2:} Next suppose $r$ is singleton on $\mathcal{L}_2$ and $r'$ is singleton on some $\mathcal{L}_k$ where $k > 2$. Note that between $\mathcal{L}_2$ and $\mathcal{L}_k$ there is no other robots (Claim~\ref{lemma: not more than two moving1 excluding L1} and Observation~\ref{observation: right half of off is  all off}). Note that all other robot on the right of $\mathcal{L}_k$ at time $t$ has color \texttt{off} and upon activation they do nothing from time $t$ onwards as they never see all robots with color \texttt{chord} on their left immediate vertical line as $r$ never moves to $\mathcal{L}_1$. In this case whenever $r'$ is activated it sees $r$ singleton on $\mathcal{L}_2$ and moves left to $\mathcal{L}_{k-1}$ (Figure~\ref{fig:moving1 move left}(b)). This way eventually $r'$ reaches $\mathcal{L}_2$ along with $r$. This is same configuration as described in case 1. So, again it would reach a contradiction.

        \begin{figure}[h]
            \centering
            \includegraphics[width=8cm]{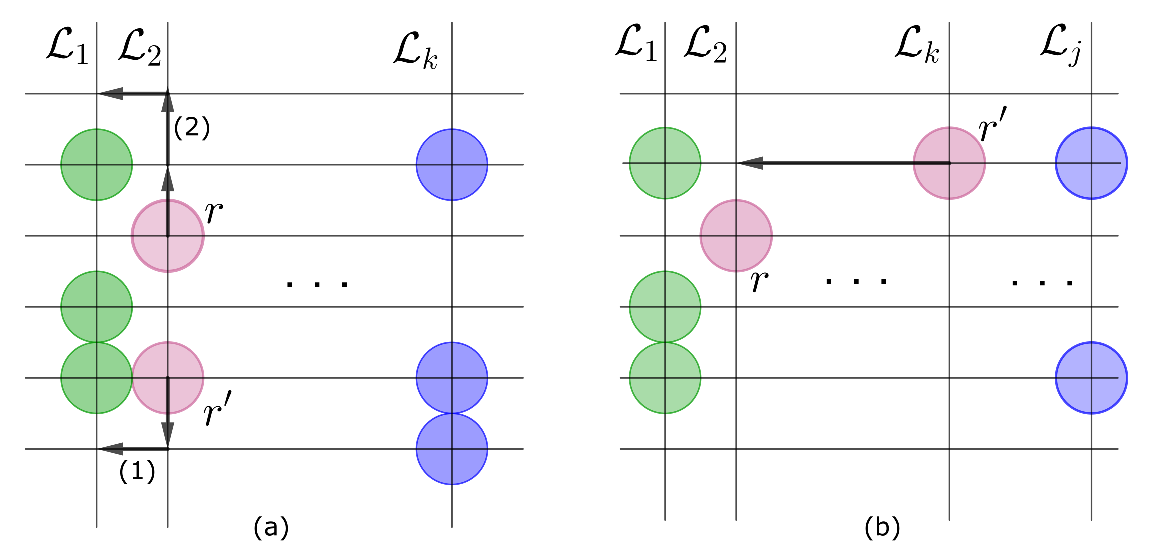}
            \caption{\textbf{(a)} $r$ and $r'$ are on $\mathcal{L}_2$ with color \texttt{moving1}. If $r$ do not move left, $r'$ moves left upon seeing all robots on $\mathcal{L}_1$ of color \texttt{chord} and in the next activation changes color to \texttt{chord}. Next $r$ moves to $\mathcal{L}_1$. \textbf{(b)} $r$ with color \texttt{moving1} is singleton on $\mathcal{L}_2$ and $r'$ with color \texttt{moving1} is also singleton on $\mathcal{L}_k$ ($k>2$). If $r$ do not move to $\mathcal{L}_1$, $r'$ moves and reaches $\mathcal{L}_2$ and converts to the case in (a).}
            \label{fig:moving1 move left}
        \end{figure}

        \textit{Case 3: } Let, $r$ is singleton on $\mathcal{L}_2$ and $r'$ is not singleton on $\mathcal{L}_k$ at time $t$ where $k > 2$. Then if $r'$ has a pending move, it moves to $\mathcal{L}_{k-1}$, otherwise if $r'$ activates after time $t$, it does nothing. For the former case, after $r'$ moves to $\mathcal{L}_{k-1}$, it  either falls into the case 1 or, the case 2. Now, for the later case, if $r'$ does nothing on $\mathcal{L}_k$ then $r$ always remains singleton on $\mathcal{L}_2$ and no new robot moves onto $\mathcal{L}_1$ from time $t$ on wards. So, at time $t$ if not all robots on $\mathcal{L}_1$ is of color \texttt{chord}, eventually they all becomes of color \texttt{chord (Claim~\ref{lemma: robot moving1 changes color to chord})}. So, upon next activation $r$ must move to $\mathcal{L}_1$. This is again a contradiction to our assumption. 

        \textit{Case 4:} In this case, let us assume $r$ is only robot in the configuration  of color \texttt{moving1} on $\mathcal{L}_2$ at time $t$. If from $t$ onwards no new robot changes its color to \texttt{moving1} then eventually there will be a time, say $t' > t$ when all robots on $\mathcal{L}_1$ have color \texttt{chord}. Also no new robot moves to $\mathcal{L}_1$ after $t'$. Thus, when $r$ activates next it must move to $\mathcal{L}_1$. So, let us now assume there is a time when a new robot $r'$ changes its color to \texttt{moving1}. Then as described in the previous cases, $r$ will move to $\mathcal{L}_1$ contradicting our assumption. 
        So, There will be a time when $r$ moves to $\mathcal{L}_1$ from $\mathcal{L}_2$.

        \textbf{Hypothesis:} For some $j >2$ and for any $i \le j$, a robot having a color \texttt{moving1} on $\mathcal{L}_i$ moves to $\mathcal{L}_{i-1}$ where $i \ge 2$.

        \textbf{Inductive step:} Let $r$ be a robot on $\mathcal{L}_{j+1}$ with color \texttt{moving1} at a time $t$, where $j \ge 2$. Now  by Claim~\ref{lemma: not more than two moving1 excluding L1}, there can be at most another robot, say $r'$ with color \texttt{moving1} on $\mathcal{L}_k$ at time $t$ ($k >1$). If possible let $r$ never move to $\mathcal{L}_{j}$. Now there are two cases.

        \textit{Case 1:} For the first case let us assume $k \le j+1$ (Figure~\ref{fig:moving1 moves left induction}(a)). First assume that $r$ is singleton on $\mathcal{L}_{j+1}$ at time $t$. Then $k < j+1$. By Observation~\ref{observation: right half of off is  all off} between $\mathcal{L}_1$ and $\mathcal{L}_{j+1}$ there are no other robot except $r'$ on $\mathcal{L}_k$ at time $t$. Note that by induction hypothesis $r'$ eventually moves to $\mathcal{L}_1$ and changes its color to \texttt{chord}, say at a time $t'>t$. Before that no robot on $H_R^O(r)$ does anything even if they are activated as they can not change their color from \texttt{off} to \texttt{moving1} due to the fact that they can never see $\mathcal{L}_1$ as their left immediate vertical line due to $r$ being on $\mathcal{L}_{j+1}$. So, after $t'$, whenever $r$ activates it sees $\mathcal{L}_I(r)=\mathcal{L}_1$ where all robots have color \texttt{chord}. Thus eventually it moves left to $\mathcal{L}_{j}$. Also, before $t'$ if $r$ is activated and sees $r'$ on some $\mathcal{L}_{k'}$ where $1<k'\le k$ then it moves left to $\mathcal{L}_j$ contrary to the assumption. So, let $r$ is not singleton on $\mathcal{L}_{j+1}$ at time $t$. Let there are $p$ robots on $\mathcal{L}_{j+1}$ at time $t$ (Figure~\ref{fig:moving1 moves left induction}(a)). Now if $r'$ is on some $\mathcal{L}_k$ where $1<k <j+1$ then, it must be singleton on $\mathcal{L}_k$. Now, until $r'$ reaches $\mathcal{L}_1$ and changes its color to \texttt{chord}, no other robot on the right of $r'$ does anything. By induction hypothesis, it can be ensured that $r'$ will reach $\mathcal{L}_1$ and eventually all robot on it would change their color to \texttt{chord} at a time say, $t'$. Now since $p> 1$, there must exist another robot, say $r_1$, having color \texttt{off} on $\mathcal{L}_{j+1}$ which is terminal on it at time $t'$. Also at time $t'$ there is no other robot in between $\mathcal{L}_1$ and $\mathcal{L}_{j+1}$ and no other robot except $r_1$ moves in between $\mathcal{L}_1$ and $\mathcal{L}_{j+1}$ from time $t'$ onwards unless $r_1$ reaches $\mathcal{L}_1$ and changes its color to \texttt{chord}. This is because only $r_1$ can change its color to \texttt{moving1} from \texttt{off}, after time $t'$ and until it reaches $\mathcal{L}_1$ and changes color to \texttt{chord} (by Claim~\ref{lemma: not more than two moving1 excluding L1}). This implies, after $t'$ whenever $r_1$ is activated it will change its color to \texttt{moving1} from \texttt{off} and will eventually move to $\mathcal{L}_j$ seeing all robots on $\mathcal{L}_1$ having color \texttt{chord}. Now by the hypothesis $r_1$ will reach $\mathcal{L}_1$ and change its color to \texttt{chord} eventually at a time, say $t_1$. Now at time $t_1$, $\mathcal{L}_{j+1}$ has $p-1$ robots. if $p-1 = 1$ then $r$ becomes singleton on $\mathcal{L}_{j+1}$ and as described above it will eventually move to $\mathcal{L}_j$. otherwise there always will be a robot on $\mathcal{L}_{j+1}$ which is terminal and has color \texttt{off}. For this case the terminal robot will eventually reach $\mathcal{L}_1$ as described above and will change its color to \texttt{chord}. This implies eventually number of robots on $\mathcal{L}_{j+1}$ will decrease until only $r$ remains. For this case, $r$ moves to $\mathcal{L}_j$ eventually, contrary to the assumption.

        \textit{Case 2:} Let $k > j+1$ (Figure~\ref{fig:moving1 moves left induction}(b)). This implies $ r$ is singleton on $\mathcal{L}_{j+1}$ also, there are no robots between $\mathcal{L}_1$ and $\mathcal{L}_k$ except $r$. First assume $r'$ is singleton on $\mathcal{L}_k$ (Figure~\ref{fig:moving1 moves left induction}(b)). Then upon seeing only $r$ on $\mathcal{L}_I(r')$ it will first move left until it reaches $\mathcal{L}_{j+1}$. Now as argument-ed for the previous case 1, $r'$ will reach $\mathcal{L}_1$ eventually and change its color to \texttt{chord}. Note that since $r$ is  singleton on $\mathcal{L}_{j+1}$, no robot from $H_R^O(r)$ does anything upon activation. Also, there are no robots in between $\mathcal{L}_1$ and $\mathcal{L}_{j+1}$ and all robots on $\mathcal{L}_1$ has color \texttt{chord}. Note that now when $r$ will be activated it will always see $\mathcal{L}_I(r)=\mathcal{L}_1$ where all robots have color \texttt{chord} and thus eventually it will move to $\mathcal{L}_j$. This is again a contradiction. Thus, let us assume at time $t$, $r'$ is not singleton on $\mathcal{L}_k$. For this case $r'$ will not do anything even if it is activated. Thus, eventually $r$ will see all robots with color \texttt{chord} on  $\mathcal{L}_I(r)$ (i.e., $\mathcal{L}_1$). Thus $r$ will move to $\mathcal{L}_j$.

\begin{figure}[h]
    \centering
    \includegraphics[width=10cm]{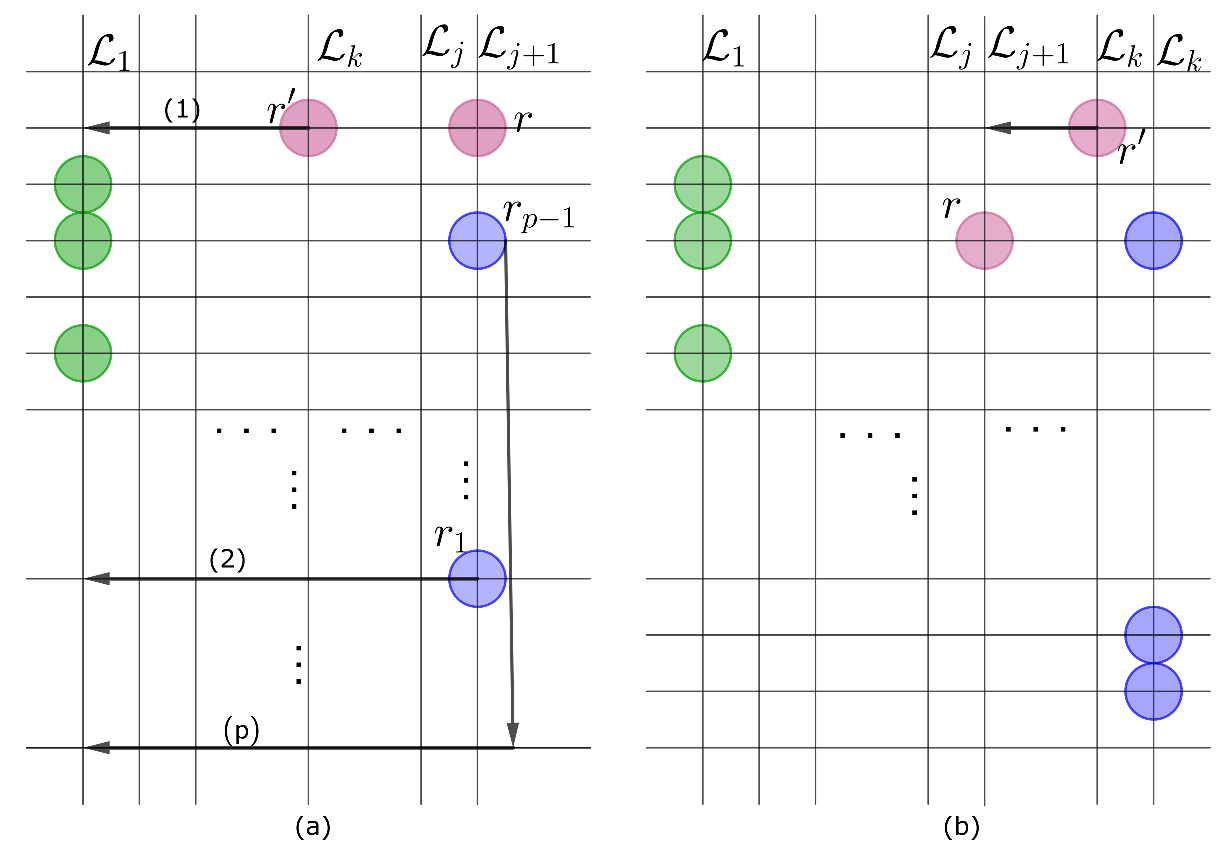}
    \caption{\textbf{(a)} $r'$ moves to $\mathcal{L}_1$ first and changes color to \texttt{chord}. Next $r_1$ changes color to \texttt{moving1} and do same as $r'$. Eventually $r_{p-1}$ also moves to $\mathcal{L}_1$ and changes color to \texttt{chord}. Next whenever $r$ activates it moves and eventually reaches $\mathcal{L}_j$. \textbf{(b)} Here if $r$ does not move, singleton $r'$ moves left. After reaching $\mathcal{L}_{j+1}$ it transforms into case I.}
    \label{fig:moving1 moves left induction}
\end{figure}
        So considering another robot, $r'$ having color \texttt{moving1} on $\mathcal{L}_k$ ($k >1$) at time $t$ we always reach a contradiction. So, let at time $t$, $r$ is the only robot with color \texttt{moving1} on $\mathcal{L}_{j+1}$ ($j \ge 2$). For some pending move, another robot may change its color to \texttt{moving1} later, say at $t' > t$. Now before $t'$ if $\mathcal{L}_1$ has all robot with color \texttt{chord} and $r$ is activated, it will move  to $\mathcal{L}_{j}$. Otherwise, we reach the same configuration described in case 1 and case 2. So, for this case also, we have a contradiction. Thus, $r$ will eventually move to $\mathcal{L}_j$. 
    \end{proof}

\begin{lemma}
    \label{lemma: off changes to moving11=}
    After $\mathcal{L}_1$ is fixed, let $\mathcal{L}_k$ be the first vertical line with a robot $r$ of color \texttt{off} at a time $t$. Then $r$ eventually changes its color to \texttt{moving1} in Phase 1.
\end{lemma}
\begin{proof}
   Let after $\mathcal{L}_1$ is fixed, $t$ be a time when $\mathcal{L}_k$ is the first vertical line that contains a robot $r$ having color \texttt{off}.This implies, all robots between $\mathcal{L}_1$ and $\mathcal{L}_k$ (if exists) has color \texttt{moving1} at time $t$. Now let from $t$ onwards $r$ never changes its color to \texttt{moving1}. Let there be $p \ge 1$ robots at time $t$ on $\mathcal{L}_k$. Now we have two cases.

   \textit{Case 1:} Let $p =1$. That is $r$ is singleton on $\mathcal{L}_k$ at time $t$ (Figure~\ref{fig:off to moving1}(a)). Then robots on $H_R^O(r)$ must be of the color \texttt{off} at time $t$ (Observation~\ref{observation: right half of off is  all off}) and does nothing upon activation from time $t$ onwards as $r$ does not change its color. Now, by lemma~\ref{lemma:moving1 on L_j moves to L_(j-1)},  all robots in between $\mathcal{L}_1$ and $\mathcal{L}_k$ eventually moves to $\mathcal{L}_1$ and change their color to \texttt{chord} at a time, say $t_0$. And within this time no new robot changes its color to \texttt{moving1} from \texttt{off}. Thus from $t_0$ onwards, if $r$ does not change its color, the configuration remains unchanged. Now, when $r$ is activated after time $t_0$, it sees all robots on $\mathcal{L}_I(r) = \mathcal{L}_1$ has color \texttt{chord} and thus changes its color to \texttt{moving1}, contrary to our assumption.

\begin{figure}[h]
    \centering
    \includegraphics[width=9cm]{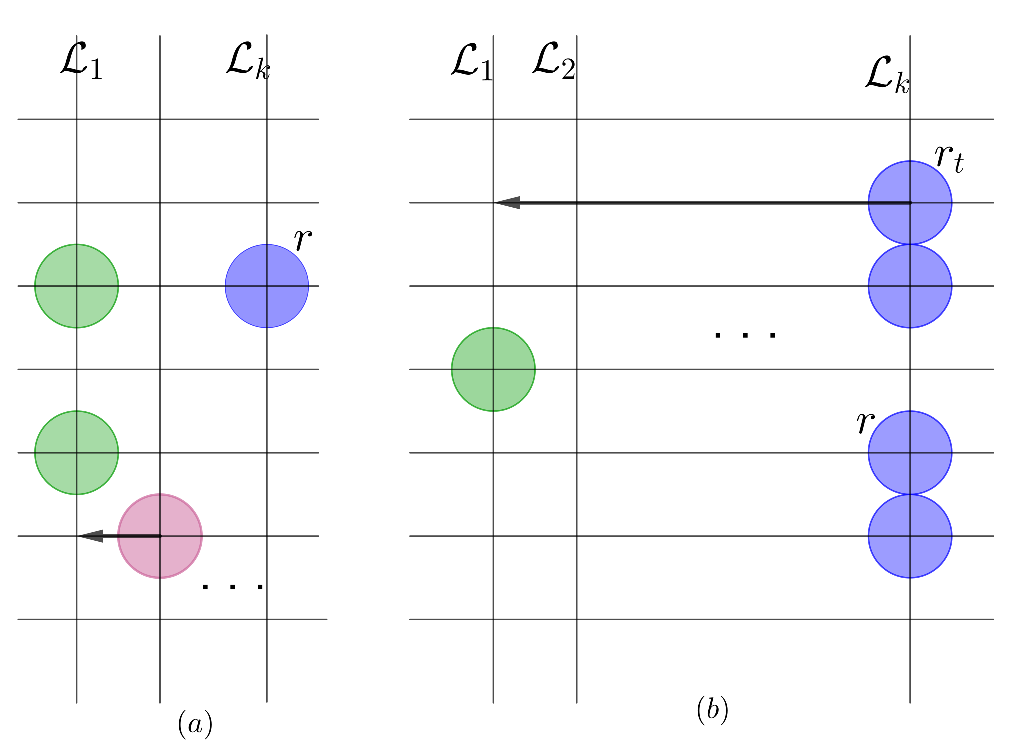}
    \caption{\textbf{(a)} $r$ with color \texttt{off} is singleton on $\mathcal{L}_k$ the robots between $\mathcal{L}_1$ and $\mathcal{L}_k$ moves to $\mathcal{L}_1$ and changes color to \texttt{chord}. Next, whenever $r$ activates it changes color to \texttt{moving1}. \textbf{(b)} $r$ is not singleton on $\mathcal{L}_k$ . All robots between $\mathcal{L}_1$ and $\mathcal{L}_k$ has already moved to $\mathcal{L}_1$ and changed color to \texttt{chord}. Next, whenever $r_t$ activates it changes color \texttt{moving1} and moves left decreasing the number of robots on $\mathcal{L}_k$ and eventually transforming into the case shown in (a).} 
    \label{fig:off to moving1}
\end{figure}

   \textit{Case 2:} Let $p >1$ (Figure~\ref{fig:off to moving1}(b)). For this case, we will show that number of robots on $\mathcal{L}_k$ eventually decreases until it becomes one. Which we already discussed in case 1. First observe that no robot on $H_R^O(r)$ at time $t$, moves to $\mathcal{L}_k$  as they will remain of color \texttt{off} (observation~\ref{observation: right half of off is  all off}) from time $t$ onwards if $r$ does not change its color to \texttt{moving1}. So the number of robots on $\mathcal{L}_k$ never increase from time $t$ onwards. Now let the number of robots on $\mathcal{L}_k$ never decrease i.e., it remains the same. This implies, no robot from $\mathcal{L}_k$ changes its color to \texttt{moving1} from time $t$ onwards. Now by Lemma~\ref{lemma:moving1 on L_j moves to L_(j-1)} all robots in between $\mathcal{L}_1$ and $\mathcal{L}_k$ have color \texttt{moving1} at time $t$ and they will eventually move to $\mathcal{L}_1$ and change their color to \texttt{chord} at a time, say $t_1 \ge t$. Now let $r_t$ be the terminal robot on $\mathcal{L}_k$ which is activated first after time $t_1$. Now upon activation, it must see all robots on $\mathcal{L}_I(r_t)= \mathcal{L}_1$ has color \texttt{chord} and thus changes its color to \texttt{moving1} contrary to our assumption. Thus the number of robots on $\mathcal{L}_k$ will decrease and eventually have only $r$. So according to case 1, we will again reach a contradiction.
   As for both cases contradictions are achieved, our assumption that $r$ never changes its color to \texttt{moving1} from time $t$ onwards is false. So there is a time when $r$ changes its color to \texttt{moving1}.
 \end{proof}

So using Lemma~\ref{lemma: L1 is fixed}
we guarantee that within a finite time after the first move by any robot from the initial configuration $\mathcal{L}_1$ will be fixed.  Let $\mathcal{L}_k$ be the first vertical line on the right of $\mathcal{L}_1$ where a robot of color \texttt{off} exists ($k > 2$). So all robots in between $\mathcal{L}_1$ and $\mathcal{L}_k$ has color \texttt{moving1}. Thus by Lemma~\ref{lemma:moving1 on L_j moves to L_(j-1)} eventually all these robots will move to $\mathcal{L}_1$. Now observe that all robots on $\mathcal{L}_k$ will have color either \texttt{moving1} or \texttt{off}. Now using Lemma~\ref{lemma:moving1 on L_j moves to L_(j-1)} it can  be said that all robots of color \texttt{moving1} on $\mathcal{L}_k$ will move to $\mathcal{L}_1$ eventually. Also by Lemma~\ref{lemma: off changes to moving11=} all robots of color \texttt{off} on $\mathcal{L}_k$ change their color to \texttt{moving1} and move to $\mathcal{L}_1$ eventually. Thus  Within finite time all robots of $\mathcal{L}_k$ move to $\mathcal{L}_1$. Now if there are no other robots on the right of $\mathcal{L}_k$ then we are done. Otherwise the first vertical line on the right of $\mathcal{L}_1$ containing a robot of color \texttt{off}, shifts right.  Eventually, there will be one such line that does not have any other robot on its right and all other robots will be on $\mathcal{L}_1$. Now as described above all other robots of that line will also eventually move to $\mathcal{L}_1$. In this moment all non terminal robots on $\mathcal{L}_1$ will have color \texttt{chord} and the terminal robots either will have color \texttt{chord} or \texttt{moving1} (Claim~\ref{lemma: moving1 terminal on L1 in phase 1}). From this above discussion, we can have the following theorem.

\begin{theorem}
    \label{Thm: Phase1 line}
    There exists a time $t$  when all robots move to a single line with non terminal robots having color \texttt{chord} and terminal robots having color either \texttt{chord} or, \texttt{moving1} by executing the algorithm~\ref{Algo_Phase1}: Phase 1 from any initial configuration assuming one axis agreement, under asynchronous scheduler. 
\end{theorem}
We now proof the following theorem which states that A configuration where all robots are on a single line with all non terminal robots have color \texttt{chord} and terminal robots have color either \texttt{chord} or, \texttt{moving1} will eventually change into a Phase 1 Final Configuration. This theorem ensures the termination of Phase 1.

\begin{theorem}
    \label{thm: phase1 final config}
    Let in a configuration $\mathcal{C},$ all robots are on a single line where each non terminal robots have color \texttt{chord} and the terminal robots have color either \texttt{chord} or \texttt{moving1}. Then In finite time the configuration will change into a Phase 1 Final Configuration (P1FC).
\end{theorem}
\begin{proof}
    By the theorem~\ref{Thm: Phase1 line}, there exists a time $t$ when all robots will be on a single line. All robots on that line which are not terminal must have color \texttt{chord} at time $t$ and the terminal robots can have color either \texttt{moving1} or, \texttt{chord} at time $t$. Let both terminal robots, say $r_1$ and $r_2$, has color \texttt{moving1} at time $t$. Then among $r_1$ and $r_2$, whichever is activated first, say $r_1$ without loss of generality, must see a robot of color \texttt{chord} on $\mathcal{L}_V(r_1) = \mathcal{L}_1$ and changes its color to \texttt{chord}.  The guarantee that $r_1$ will see a robot of color \texttt{chord} comes from the fact that the non terminal robots with color \texttt{chord} on $\mathcal{L}_1$ do not move out of $\mathcal{L}_V(r_1)$ until they see a robot with color \texttt{diameter} on $\mathcal{L}_V(r_1)$.   So there must exist a time $t_1 >t$ when all robots are on $\mathcal{L}_1$ and all but at most one terminal robot say $r_2$ has color \texttt{chord} and  $r_2$ has either color \texttt{chord} or \texttt{moving1}. Now there are two cases. In the first case, we assume all robots have color \texttt{chord} at $t_1$, and in the second case we assume all robots except $r_2$ have color \texttt{chord} and $r_2$ has color \texttt{moving1}.

\textit{Case I:} Let us consider the case when all robots on $\mathcal{L}_1$ has color \texttt{chord} at time $t_1$. Now after $t_1$ whichever robot of $r_1$ and $r_2$ is activated first, changes its color to \texttt{diameter}. Note that both terminal robots can have color \texttt{diameter} too. So, there exists a time  $t_2 >t_1$ when all robots are on $\mathcal{L}_1$, all non terminal robots have color \texttt{chord} and at least one terminal robot has color \texttt{diameter}. 

\textit{Case I(a):} Now  if both the terminal robots $r_1$ and $r_2$ have color \texttt{diameter} at time $t_2$ (Figure~\ref{fig:P1FC Formation}(a)), then the non terminal robots on $\mathcal{L}_1$ that can see  $r_1$ or $r_2$ with color \texttt{diameter} moves either left or right once by executing \textsc{ChordMove} subroutine. After a robot moves to the left or right of $\mathcal{L}_V(r_1)$ by executing \textsc{ChordMove}, the next non terminal robot on $\mathcal{L}_V(r_1)$ can now see at least one robot of color \texttt{diameter} and thus execute the \textsc{ChordMove} subroutine. This way all non terminal robots on $\mathcal{L}_V(r_1)$ will execute \textsc{Chordmove} at least once. Note that a robot that has executed \textsc{ChordMove} once does not move again until Phase 1 Final configuration is achieved. This is because the robot does not execute Phase 1 as it will see a robot of color \texttt{diameter}  which is not on its own vertical line. Also, it will not execute Phase 2 as it never sees only two robots with color \texttt{diameter} on its right immediate (or, left immediate) vertical line until Phase 1 Final Configuration is achieved. Now Let $t_3$ be a time the last non terminal robot on $\mathcal{L}_V(r_1)$  moved right or left after executing \textsc{ChordMove}. Then we can ensure that only $r_1$ and $r_2$ are on $\mathcal{L}_2$ with color \texttt{diameter} at time $t_3$ and all other robots are on $\mathcal{L}_1$ and $\mathcal{L}_3$ with color \texttt{chord} and also they are strictly between $\mathcal{L}_H(r_1)$ and $\mathcal{L}_H(r_2)$. Hence at $t_3$, the configuration becomes a Phase 1 Final Configuration. 

\textit{Case I(b):} Now let at time $t_2$, only $r_1$ has color \texttt{diameter} and $r_2$ has color \texttt{chord} (Figure~\ref{fig:P1FC Formation}(b)). Now if $r_2$ is activated before any of the non terminal robots move by executing \textsc{ChordMove} it changes its color to \texttt{diameter}. For this case, we can show that eventually Phase 1 Final Configuration will be achieved (using a similar argument as  Case I(a)). So, let before $r_2$ is activated a robot which is not terminal on $\mathcal{L}_1$  and sees $r_1$ executes \textsc{ChordMove} and moves left. Then $r_2$ never changes its color to \texttt{diameter} until all non terminal robots move left or right of $\mathcal{L}_V(r_1) = \mathcal{L}_V(r_2)$ by executing \textsc{ChordMove}. When all such robots move, $r_2$ then can see $r_1$ with color \texttt{diameter} and then it changes its color to \texttt{diameter}. Note that after $r_2$ changes its color \texttt{diameter} the configuration becomes a Phase 1 Final Configuration. 

\textit{Case II: } Let at time $t_1$ all but one terminal robot, say $r_2$ (without loss of generality), on $\mathcal{L}_1$ has color \texttt{chord}. Let $r_2$ has color \texttt{moving1} at time $t_1$. Now, if $r_1$ is activated and sees a robot of color \texttt{chord} on its vertical line and sees no robot on both of its left and right open halves then $r_1$ changes its color to \texttt{diameter}. Let this happens at a time $t_4 > t_1$.  Then similar to the above discussion all non terminal robots on $\mathcal{L}_V(r_1)$, that see $r_1$  after time $t_4$, on their own vertical line executes \textsc{ChordMove} and move left or right. If at least one such non terminal robot has already executed \textsc{ChordMove} then even if $r_2$ is activated after that, with color \texttt{moving1} it does not change its color as it sees its left open half non-empty. 

\textit{Case II(a):} If $r_2$ is activated before any robot moves from $\mathcal{L}_V(r_1) = \mathcal{L}_V(r_2)$ then it sees a robot with color \texttt{chord} on $\mathcal{L}_V(r_2)$ and sees $H_L^O(r_2)$ empty. Thus $r_2$ in this case changes its color to \texttt{chord} at a time say $t_5 > t_1$. In the configuration, if $t_5 > t_4$ then at time $t_5$, all robots have color \texttt{chord} except $r_1$, which has color \texttt{diameter}. This is similar to the case I(b). Thus for this case eventually the configuration will become a Phase 1 Final Configuration. So, let $t_5 \le t_4$. If $t_5 < t_4$, then at $t_5$, all robots will be on a single line with color \texttt{chord}. This is similar to case I and thus eventually the configuration will change to a Phase 1 Final Configuration. Now if $t_5 =t_4$, then at $t_5$, $r_1$ will have color \texttt{diameter} and $r_2$ will have color \texttt{chord} which is similar to the case I(b). Thus again the configuration will eventually become a Phase 1 Final Configuration. 

\textit{Case II(b):} Now let after $t_1$, $r_2$ is activated for the first time when at least one non terminal robot already executed \textsc{ChordMove} (Figure~\ref{fig:P1FC Formation}(c)). Then $r_2$ does not change its color even if it is activated as it does not have $H_L^O(r_2)$ non-empty. Now When all non terminal robots execute \textsc{ChordMove} once then on $\mathcal{L}_V(r_2)$ there are only two robots $r_1$ with color \texttt{diameter} and $r_2$ with color \texttt{moving1}. Now when $r_2$ is activated again, it sees $r_1$ with color \texttt{diameter} on $\mathcal{L}_V(r_2)$ and changes its color to \texttt{diameter}. After this, the configuration again becomes a Phase 1 Final Configuration.

For all of the above cases, it can be ensured that A configuration where all robots are on a single line with all non terminal robots having color \texttt{chord} and terminal robots having color either \texttt{chord} or, \texttt{moving1} will change into a Phase 1 Final Configuration within finite time.
\end{proof}
\begin{figure}
    \centering
    \includegraphics[width=10cm]{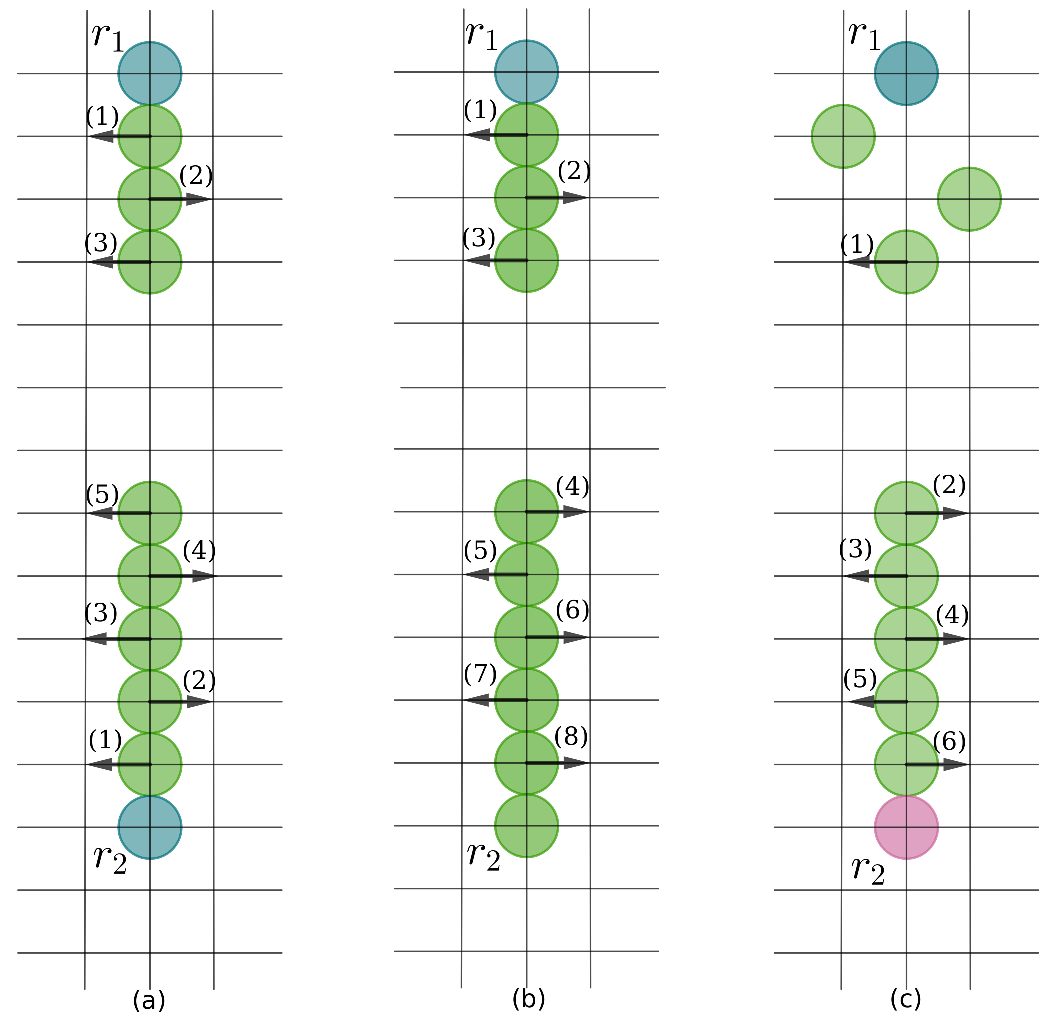}
    \caption{\textbf{(a)} Both terminal robots have color \texttt{diameter} in $\mathcal{C}(t_2)$. \textbf{(b)} $r_1$ has color \texttt{diameter} and $r_2$ has color \texttt{chord} in $\mathcal{C}(t_2)$. $r_2$ does not change color \texttt{diameter} until all non terminal robots of color \texttt{chord} execute \textsc{ChordMove} exactly once and moves to left or right. \textbf{(c)} $r_2$ still has color \texttt{moving1} while some non terminal robots from the line already executed \textsc{ChordMove}. In this case also $r_2$ changes color to \texttt{diameter} when all non terminal robots moves left or right. The numbers in the bracket denotes the order in which the robots move.}
    \label{fig:P1FC Formation}
\end{figure}
\subsection{Phase 2}
Phase 1 terminates when the configuration is a P1FC. So, the initial configuration in Phase 2 is a P1FC where $r_1$ and $r_2$ be the only robots on $\mathcal{L}_2$ and they have color \texttt{diameter}. Note that these two robots are already terminated in Phase 1 after changing color to \texttt{diameter}. Now, these two robots will help the other robots to agree on the circle to be formed. The line segment of $\mathcal{L}_V(r_1)$ between $r_1$ and $r_2$  will be agreed by other robots as a diameter of the circle to be formed. So, if a robot sees both $r_1$ and $r_2$ it knows the circle, say $\mathcal{CIR}$.

Now, observe that in a P1FC, each horizontal line between $\mathcal{L}_H(r_1)$ and $\mathcal{L}_H(r_2)$ contains at most one robot. Now each of these horizontal lines has exactly two grid points on the circumference of the circle $\mathcal{CIR}$, one on the left of the agreed diameter and the other on the right. So, for all $r \in \mathfrak{R}$, if $r$ terminates at a grid point on $\mathcal{L}_H(r)$  which is on the circumference of $\mathcal{CIR}$, we say that circle is formed. 
In Phase 2 of algorithm $CF\_FAT\_GRID$, any robot $r$ terminates at the grid point on $\mathcal{L}_H(r)$ which is on the circumference of $\mathcal{CIR}$ and on the left (resp. right) of the agreed diameter if in the P1FC  $r$ is on the left (resp. right) of the agreed diameter.
\subsubsection{Brief Description of Phase 2}
Phase 2 starts when there are exactly two robots with color \texttt{diameter}, say $r_1$ and $r_2$ on the same vertical line, and all other robots have color \texttt{chord} on $\mathcal{L}_I(r_1)$ and $\mathcal{R}_I(r_1)$. Difference between number of robots on  $\mathcal{L}_I(r_1)$ and $\mathcal{R}_I(r_1)$ is at most two. Note that the vertical line $\mathcal{L}_V(r_1)$ divides the circle on two halfs $H_L^O(r_1)$ and $H_R^O(r_1)$. Here, we describe Phase 2 for the robots only in $H_L^O(r_1)$. The algorithm for the robots in $H_R^O(r_1)$ will be similar. 

Note that due to the procedure \textsc{ChordMove} in Phase 1, all robots with color \texttt{chord} must be strictly between the horizontal lines passing through $r_1$ and $r_2$. Observe that, in Phase 1, if a robot $r$ of color \texttt{chord} sees at least one robot of color \texttt{diameter} then there are two possibilities. Either $r$ sees  at least one robot of color \texttt{diameter}, say $r_1$, on $\mathcal{L}_V(r)$ or, it sees $r_1$ on $\mathcal{R}_I(r)$ (resp. $\mathcal{L}_I(r)$) while there must be another robot of color \texttt{chord} on $\mathcal{R}_I(r)$ (resp. $\mathcal{L}_I(r)$). In Phase 2 we have shown that (Lemma~\ref{lemma: decide phase 2 chord}) a robot $r$ of color \texttt{chord} always sees at least one robot of color \texttt{diameter} which can not be on $\mathcal{L}_V(r)$ (as $r$ never moves to the vertical line where the robots of color \texttt{diameter} are located). Now if $r$ sees  a robot $r_1$ of color \texttt{diameter} on $\mathcal{R}_I(r)$ (resp.$\mathcal{L}_I(r)$) then, it must see the other robot $r_2$ of color \texttt{diameter} on $\mathcal{R}_I(r)$ (resp. $\mathcal{L}_I(r)$) too. But the difference from Phase 1 is that, here $\mathcal{L}_V(r_1)$ does not have any other robots of color \texttt{chord}. Thus a robot of color \texttt{chord} can always distinguish between Phase 1 and Phase 2.

Now, a robot which is terminal on $\mathcal{L}_I(r_1)$, say $r$, must see both the robots $r_1$ and $r_2$. In this case, it moves toward its left. Before this move $r$ changes its color to \texttt{off} only if the horizontal distance of $r$ from $r_1$ or $r_2$ is $\lceil \frac{d}{2}\rceil -1$. Observe that due to this rule after a finite time all robots which were initially on $\mathcal{L}_I(r_1)$ with color \texttt{chord}, reach a vertical line which is $\lceil\frac{d}{2}\rceil$ distance away from $r_1$ and $r_2$ with color \texttt{off}. Now we claim that a robot with color \texttt{off} always sees at least one of $r_1$ or $r_2$ (Lemma~\ref{lemma:off decide phase 2}). By this condition, a robot with color \texttt{off}, can identify whether it is in Phase 1 or Phase 2.  Let us name the vertical line $\mathcal{L}_V(r_1)$ as $v_0$, and $v_i$ be the $i$-th vertical line on the left of $v_0$. So, after a finite time, all robots will be on $v_{\lceil \frac{d}{2} \rceil}$ strictly between $\mathcal{L}_H(r_1)$ and $\mathcal{L}_H(r_2)$ having color \texttt{off}. In this configuration, all the robots can see both $r_1$ and $r_2$ and thus can calculate the point $c$ which is equidistant from both $r_1$ and $r_2$ on $v_0$. There can be at most two robots on $v_{\lceil \frac{d}{2} \rceil}$ which are nearest to $c$. Also, if there are two such robots then there can be no other robots strictly between them on $v_{\lceil \frac{d}{2} \rceil}$. Thus a nearest robot to $c$ on $v_{\lceil \frac{d}{2} \rceil}$ can understand if it is nearest to $c$. Such a robot, say $r$,  first moves right after changing the color to \texttt{moving1} and does not move further until it sees no other robots on its left. Now , after $r$ moves, another robot, say $r'$, on $v_{\lceil \frac{d}{2} \rceil}$ moves right only if it sees $r$ with color \texttt{moving1} on its immediate right vertical line $v_{\lceil \frac{d}{2} \rceil-1}$ and sees no other robot on the left of $v_0$ strictly between $\mathcal{L}_H(r)$ and $\mathcal{L}_H(r')$. Thus there exists a time when all robots on the left of $v_0$ are on $v_{\lceil \frac{d}{2} \rceil-1}$ and have color \texttt{moving1}. We call this configuration a $(\lceil\frac{d}{2}\rceil-1)-$\textit{Left Sub Circle Configuration  ($(\lceil\frac{d}{2}\rceil-1)-$LSCC)}. More formally,
\begin{definition}[$j-$ Left Sub Circle Configuration Left ($j$-LSCC)]
    A configuration where $r_1$ and $r_2$ with color \texttt{diameter} be the only two robots on $v_0$ is called a $j$- Left sub circle configuration (Figure~\ref{fig:jLSCC}) if 
    \begin{enumerate}
        \item All robots of color \texttt{moving1} on the left of $v_0$ are the only robots on vertical line $v_j$.
        \item There are no robots strictly between $v_j$ and $v_0$.
        \item All robots on the left of $v_0$ are strictly between $\mathcal{L}_H(r_1)$ and $\mathcal{L}_H(r_2)$
        and each horizontal line contains at most one robot.
        \item All robots on left of $v_j$ (if any) must be of color \texttt{done}
    \end{enumerate}
\end{definition}

\begin{figure}[h]
    \centering
    \includegraphics[width=3.5cm]{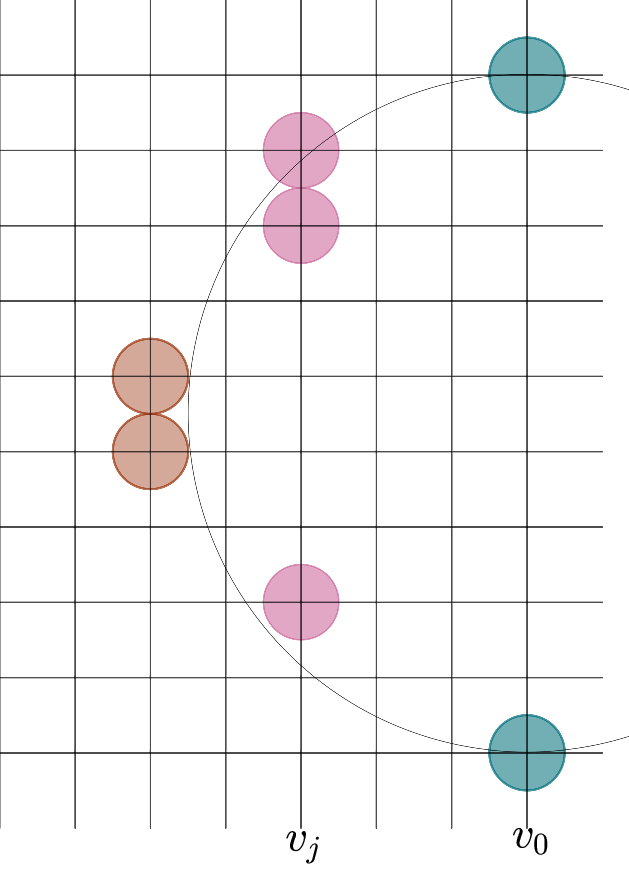}
    \caption{$j-$LSCC where on $v_0$ there are only two robots of color \texttt{diameter}, all robots of color \texttt{moving1} in the configuration  are on $v_j$ and all other robots are on left of $v_j$ with color \texttt{done}.}
    \label{fig:jLSCC}
\end{figure}
Let  $r$ be a robot on the left of $v_0$ that sees both $r_1$ and $r_2$ with color \texttt{diameter}. Then it can find out the point $c$  which is equidistant from $r_1$ and $r_2$ on $v_0$. Let $\mathcal{CIR}$ be the circle with center at $c$ and radius $\frac{d}{2}$ where $d$ is the length between $r_1$ and $r_2$ along $v_0$. Let $\mathcal{L}_{\perp}(r)$ be the line through $r$ and perpendicular to $v_0$. Let $C_r$ be the point of intersection of $\mathcal{CIR}$ and $\mathcal{L}_{\perp}(r)$. We say that a robot $r$ is on $\mathcal{CIR}$ if $0 \le distance(r, C_r) <1$ and $r$ is on left of $C_r$ on $\mathcal{L}_{\perp}(r)$. 

In a $j-$ LSCC configuration we can divide the robots of color \texttt{moving1} in two classes namely, $\mathcal{IN}_j$ and $\mathcal{OUT}_j$. We say that a robot of color \texttt{moving1} in a $j-$LSCC is in $\mathcal{IN}_j$ if it is strictly inside the circle $\mathcal{CIR}$ otherwise, it is in $\mathcal{OUT}_j$ (i.e., when the robot is either strictly outside or on the circle $\mathcal{CIR}$).

Now a robot of color \texttt{moving1} can distinguish Phase 2 from Phase 1 from the fact that a robot of color \texttt{moving1} always sees at least one robot of color \texttt{diameter} not on its vertical line in Phase 2 (Lemma~\ref{lemma:decide Phase 2 moving1}). Let $r$ be a robot of color \texttt{moving1} in Phase 2. If $r$ on some $v_j$ ($j>1$), sees both $r_1$ and $r_2$ with color \texttt{diameter} on $\mathcal{R}_I(r)$ and  no robot of color \texttt{off} or \texttt{moving1} visible to $r$ on $\mathcal{L}_I(r)$, also if  $r$ is nearest to $c$, then $r$ moves left after changing its color to \texttt{done} only if $r$ is strictly inside the circle $\mathcal{CIR}$. Otherwise, if $r$ is on or strictly outside the circle then, it moves right. Now if $j=1$ then $r$ moves left after changing the color to \texttt{done} when $r$ is strictly inside the circle, nearest to $c$ and sees no robot of color \texttt{off} or \texttt{moving1} on $\mathcal{L}_I(r)$. Otherwise, if it is on the circle, it changes its color to \texttt{done} and terminates on $\mathcal{L}_V(r) =v_1$. If $r$ with color \texttt{moving1} sees another robot $r'$ of color \texttt{moving1} on its right then it only moves right when it sees there is no other robot inside the rectangle bounded by the lines $\mathcal{L}_H(r)$, $\mathcal{L}_H(r')$, $\mathcal{L}_V(r)$ and $\mathcal{L}_V(r')$.

We now have the following observations. Also, the pseudo code of Phase 2 is presented in Algorithm~\ref{algo phase2}.

\begin{figure}[]
  \centering
   \begin{minipage}{1\linewidth}
\begin{algorithm}[H]
\small
\label{algo phase2}
    \SetKwInOut{Input}{Input}
    \SetKwInOut{Output}{Output}
    \SetKwProg{Fn}{Function}{}{}
    \SetKwProg{Pr}{Procedure}{}{}

    \Pr{\textsc{Phase2()}}{

    $r \leftarrow$ myself\\
    $d=$ distance between two robots with color \texttt{diameter}\\
    robot with color \texttt{diameter} is on right (resp. left) open half of $r$\\
    $c=$ midpoint of two robots with color \texttt{diameter}


    \uIf{$r.color =$ \texttt{chord}}
         {
                \If{$r$ sees 
                  only two robots with color \texttt{diameter} on $\mathcal{R}_I(r)$ (resp. $\mathcal{L}_I(r)$)
                }
                {

         \uIf{ $r$ is terminal 
         }
            {
            \uIf{horizontal distance from robot with color \texttt{diameter} $= \lceil \frac{d}{2} \rceil-1$}
                 {
                 $r.color =$ \texttt{off}\\
                 move horizontally away from robot with color \texttt{diameter}
                  }
            \Else{ move horizontally away from robot with color \texttt{diameter}}

            }

                }
         
         }
         
    \uElseIf{$r.color =$ \texttt{off}}
            {
                \If{$r$ sees at least one robot of color \texttt{diameter} }
                {
                    \uIf{$r$ sees exactly two robots on $\mathcal{R}_I(r)$ (resp. $\mathcal{L}_I(r)$) with color \texttt{diameter} and $r$ is nearest to the center $c$ }
                    {
                            $r.color =$ \texttt{moving1}\;
                            move towards the robot of color \texttt{diameter}\;

                    }
                    \ElseIf{$r$ sees $ r'$, a robot of color \texttt{moving1} on $\mathcal{R}_I(r)$(resp. $\mathcal{L}_I(r)$) such that no robots between $\mathcal{L}_H(r)$ and $\mathcal{L}_H(r')$}
                    {   
                        $r.color = $ \texttt{moving1}\;
                        move horizontally towards the robot of color \texttt{diameter}\;
                    }
                }
            }

    \ElseIf{$r.color =$ \texttt{moving1}}
            {
               \If{$r$ sees at least one robot of color \texttt{diameter} not on $\mathcal{L}_V(r)$}
               {
                    \uIf{$\mathcal{R}_I(r)$ (resp. $\mathcal{L}_I(r)$) has exactly two robots of color \texttt{diameter}  }
                    {   
                        \eIf{$\mathcal{R}_I(r)$ (resp. $\mathcal{L}_I(r)$) is more than one hop away from $\mathcal{L}_V(r)$}
                        {
                            \If{$r$ is nearest to $c$}
                            {
                                \If{there is no robot of color \texttt{off} or \texttt{moving1} on $\mathcal{L}_I(r)$ (resp. $\mathcal{R}_I(r)$)}
                                {
                                    \eIf{$r$ is strictly inside of $\mathcal{CIR}$ }
                                    {
                                        $r.color=$ \texttt{done}\;
                                        move left (right resp.)\;
                                    }
                                    {
                                    move right (left resp.)\;
                                    }
                                }
                            }
                        }
                        {   
                            \If{$r$ is nearest to $c$}
                            {
                                \If{there is no robot of color \texttt{off} or \texttt{moving1} on $\mathcal{L}_I(r)$ (resp. $\mathcal{R}_I(r)$)}
                                {
                                    \eIf{$r$ is strictly inside of $\mathcal{CIR}$ }
                                    {
                                         $r.color=$ \texttt{done}\;
                                        move left (right resp.)\;
                                    }
                                    {
                                        $r.color=$ \texttt{done}
                                    }
                                }
                            }
                        }
                        
                    }

                    \ElseIf{$\mathcal{R}_I(r)$ ($\mathcal{L}_I(r)$ resp.) has a robot $r'$ of color \texttt{moving1} such that no robots on or strictly inside the rectangle bounded by  $\mathcal{L}_H(r)$, $\mathcal{L}_H(r')$, $\mathcal{L}_V(r)$ and $\mathcal{L}_V(r')$ except $r$ and $r'$}
                    {
                        move right (left resp.)\;     
                    }
               }
            }

        \ElseIf{$r.light =$ \texttt{done}}
        {
         terminate\;
        }

  }

    \caption{\textbf{Phase 2}}
    \label{leader selection}
\end{algorithm}
\end{minipage}
\end{figure}
\begin{observation}
\label{obs: chord m1 not together}
    In Phase 2, there can not be a configuration where two robots exists such that one is of color \texttt{chord} and another is color \texttt{moving1}.
\end{observation}
\begin{observation}
\label{obs: chord right of off}
    In Phase 2, if a configuration has a robot $r_c$ of color \texttt{chord} and a robot $r_o$ of color \texttt{off} on the left (respectively right) of $v_0$ then the $r_o$ must have to be on $H_L^C(r_c)$ (respectively $H_R^C(r_c)$).
\end{observation}
\begin{observation}
\label{obs: moving1 right of off}
    In Phase 2, if a configuration has a robot $r_m$ of color \texttt{moving1} and a robot $r_o$ of color \texttt{off} on the left (respectively right) of $v_0$ then the $r_o$ must have to be on $H_L^C(r_m)$ (respectively $H_R^C(r_m)$). 
\end{observation}

\subsubsection{Correctness of Phase 2}
To prove the correctness of Phase 2 we have to prove two things. \begin{enumerate}
    \item For a robot $r$ on the left(resp. right) of the agreed diameter i.e., $v_0$, if $r$ terminates, then it terminates on the grid point on the circumference of the agreed circle on left (resp. right) of $v_0$ which is on $\mathcal{L}_H(r)$ (Lemma~\ref{lemma: terminating grid point})
    \item Except  $r_1$ and $r_2$ (the robots of color \texttt{diameter}), all other robots terminate in Phase 2. (Lemma~\ref{lemma: all robot terminates}).

\end{enumerate}
For $r_1$ and $r_2$, they are already terminated on Phase 1 on the two endpoints of the agreed diameter, so they are also on the circle.
Now to prove these two things we have some other lemmas that will be useful for the proof. In the following, we prove these lemmas along with the two above-mentioned results and summarize the main result in Theorem~\ref{thm: final}. 
\begin{lemma}
 \label{lemma: decide phase 2 chord}
In Phase 2, a robot $r$ with color \texttt{chord}  always sees a robot of color \texttt{diameter}. 
\end{lemma}
\begin{proof}
Let $r_1$ and $r_2$ be the only two robots of color \texttt{diameter} on the vertical line $v_0$ at a time, say $t$ in Phase 2.
Let at time $t$, $r$ be a robot of color \texttt{chord} on $v_{i_1}$ on the left of $v_0$ that can not see both $r_1$ and $r_2$. Then there must exist two robots $r_1'$ and $r_2'$ strictly inside the rectangles bounded by $v_{i_1}$, $\mathcal{L}_H(r)$, $\mathcal{L}_H(r_1)$, $v_0$ and $v_{i_1}$, $\mathcal{L}_H(r)$, $\mathcal{L}_H(r_2)$, $v_0$ respectively (Figure~\ref{fig:Chord sees at least a diameter}). Note that $r_1'$ and $r_2'$ must have color \texttt{chord} (due to Observation~\ref{obs: chord m1 not together} and Observation~\ref{obs: chord right of off}). Let $\mathcal{L}_V(r_1')$ is the vertical line $v_{i_2}$ and $\mathcal{L}_V(r_2')$ is the vertical line $v_{i_3}$ where $i_1 > i_2 \ge i_3 > 0$. We claim that $i_2 = i_3$. Otherwise, let at time $t$, $i_2 > i_3$. Now note that, $r$ must have moved left from $v_{i_2}$ before $t$ as $i_1 > i_2$. Thus, there must exist a time $t_1 < t$ when $r$ gets activated on $v_{i_2}$ and sees it is terminal on $v_{i_2}$ and sees $r_1$ and $r_2$ on $\mathcal{R}_I(r)$ and then moves left to $v_{i_1}$. Thus, at $t_1$,  $\mathcal{R}_I(r) = v_0$. This implies $r_2'$ must be on the left or on $v_{i_2}$ at time $t_1$. This is not possible as $r_2'$ is on $v_{i_3}$ at time $t > t_1$ and a robot of color \texttt{chord} on the left of $v_0$ never moves right. Thus at time $t$, $r_1'$ and $r_2'$ must be on same vertical line say $v_{i_2}$. Now at $t$, $r$ is on $v_{i_1}$ which is on the left of $v_{i_2}$. This implies there exists a time before $t$ when $r$ moved left from $v_{i_2}$. Thus there exists a time $t_2 < t$ when $r$ is activated on $v_{i_2}$ and sees it is terminal on $v_{i_2}$ and there is no robots between $v_{i_2}$ and $v_0$. This implies at $t_2$, $r_1'$ and $r_2'$ must be on $v_{i_2}$. So at $t_2$, $r$ can not be terminal on $v_{i_2}$. The contradiction arises because our assumption that $r$ can not see both $r_1$ and $r_2$ is false. Hence, $r$ must see at least one robot of color \texttt{diameter} in Phase 2.  

\begin{figure}[h]
    \centering
    \includegraphics[width=4cm]{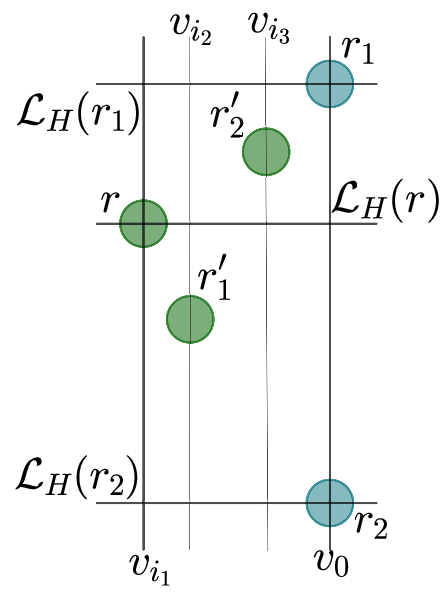}
    \caption{$r_1$ and $r_2$ is on $v_0$ with color \texttt{diameter}, $r$ is on $v_{i_1}$ with color \texttt{chord} which can not see $r_1$ and $r_2$ due to $r_1'$ on $v_{i_2}$ and $r_2'$ on $v_{i_3}$. }
    \label{fig:Chord sees at least a diameter}
\end{figure}
\end{proof}
\begin{lemma}
\label{lemma:off decide phase 2}
    In Phase 2, a robot of color \texttt{off} always sees a robot of color \texttt{diameter}.
\end{lemma}
\begin{proof}
Let $r_1$ and $r_2$ be the only two robots of color \texttt{diameter} on the vertical line $v_0$ at a time, say $t$ in Phase 2.
    Let $r$ be a robot of color \texttt{off} in Phase 2 that does not see $r_1$ and $r_2$ at $t$. So at time $t$, $r$ must be on $v_{\lceil \frac{d}{2} \rceil}$ and there must be two robots $r_1'$ and $r_2'$ strictly inside the rectangles bounded by $\mathcal{L}_H(r)$, $\mathcal{L}_H(r_1)$, $v_0, v_{\lceil \frac{d}{2} \rceil}$ and $\mathcal{L}_H(r)$, $\mathcal{L}_H(r_2)$, $v_0, v_{\lceil \frac{d}{2} \rceil}$ respectively (Figure~\ref{fig:Off sees at least one Diameter}). 
     Now, $r_1'$ and $r_2'$ both can be of color either \texttt{chord} or of color \texttt{moving1} (by Observation~\ref{obs: chord m1 not together} Observation~\ref{obs: chord right of off}, Observation~\ref{obs: moving1 right of off}). Now, if both of $r_1'$ and $r_2'$ is of color \texttt{chord}, then we reach contradiction by arguing similarly as in Lemma~\ref{lemma: decide phase 2 chord}. So, let us consider the case where both $r_1'$ and $r_2'$ are of color \texttt{moving1} at time $t$. 
    
    Let at $t$, $r$ is on $v_{\lceil \frac{d}{2}\rceil}$. Also, $r_1'$ and $r_2' $ are on  right of $v_{\lceil \frac{d}{2}\rceil}$. $\mathcal{L}_H(r)$ divides the grid in two halves. Let the half where $r_1'$ is located at time $t$ be denoted as the upper half and the half where $r_2'$ is located at $t$ be denoted as the lower half. Now there exists a time $t' < t $ such that all robots on the left of $v_0$ have color \texttt{off} and are on $v_{\lceil \frac{d}{2}\rceil}$. The first robot that moves right from $v_{\lceil \frac{d}{2}\rceil}$ after changing color to \texttt{moving1} must be nearest to $c$ at time $t'$. Let $r_{i_1}$ be such a robot. Note that $r_{i_1}$ is not $r$. So without loss of generality let it is in the upper half at $t'$. Now there can be at most another robot, say $r_{i_2}$, which is also nearest to $c$ at time $t'$. If this is the case then $r_{i_2}$ must be in the upper half at time $t'$ as otherwise $r$ becomes nearer to $c$ than $r_{i_1}$ and $r_{i_2}$ contrary to the assumption. So, the first robot, say $r_f$, from lower half that moves right from $v_{\lceil\frac{d}{2}\rceil}$ must have moved after seeing a robot, say $r_f'$ of color  \texttt{moving1} on $\mathcal{R}_I(r_f)$ such that there is no other robots strictly between $\mathcal{L}_H(r_f)$ and $\mathcal{L}_H(r_f')$ at some time $t_1$ where $t > t_1 > t' $. Note that $r_f'$ must be on the upper half at $t_1$. So, at $t_1$, the region strictly between $\mathcal{L}_H(r_f)$ and $\mathcal{L}_H(r_f')$ can not be empty as $r$ is there. So, we arrive at a contradiction due to the wrong assumption that there exists a time $t$ such that in the configuration at time $t$ there exist a robot of color \texttt{off} that can not see both $r_1$ and $r_2$ of color \texttt{diameter}. Thus any robot of color \texttt{off} always sees at least one robot of color \texttt{diameter}. 

    \begin{figure}[h]
        \centering
        \includegraphics[width=4cm]{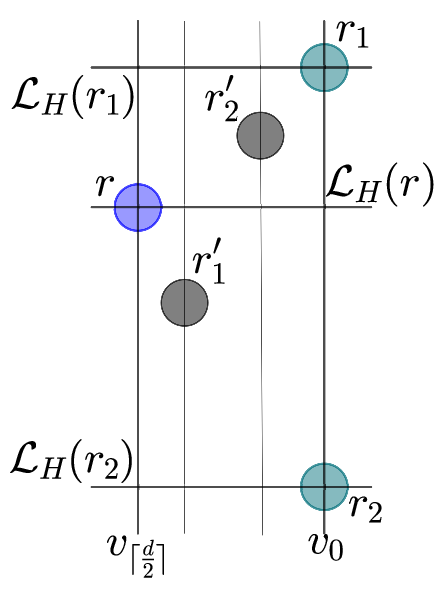}
        \caption{$r_1$ and $r_2$ is on $v_0$ with color \texttt{diameter}, $r$ is on $v_{\lceil \frac{d}{2}\rceil}$ with color \texttt{off} which can not see $r_1$ and $r_2$ due to $r_1'$ and $r_2'$.}
        \label{fig:Off sees at least one Diameter}
    \end{figure}
 \end{proof}

Now, we have to prove that a robot of color \texttt{moving1} always sees at least one robot of color \texttt{diameter} in Phase 2 which is not on its own vertical line. This will ensure that a robot of color \texttt{moving1} distinguishes Phase 2 from Phase 1. This is because in Phase 1 even if a robot of color \texttt{moving1} sees a robot of color \texttt{diameter} it must be on its own vertical line.  To prove this we have to prove the following lemma first.
\begin{lemma}
\label{lemma: master}
    Let at a time $t$, a configuration is a $j-$LSCC, where $j >1$ and $\mathcal{OUT}_j \ne \phi$. Then,
    \begin{enumerate}
        \item there exists a time $t_1 \ge t$ such that at time $t_1$, the configuration is again a $j-$LSCC  where $\mathcal{IN}_j = \phi$ and $\mathcal{OUT}_j$ at time $t_1 =   \mathcal{OUT}_j$ at time $t$.
        \item Moreover, there is another time $t_2 > t_1$ such that at $t_2$, the configuration is a $(j-1)-$LSCC and $\mathcal{IN}_{j-1} \cup \mathcal{OUT}_{j-1}$ at time $t_2 = \mathcal{OUT}_j$ at time $t_1$. 
    \end{enumerate}
\end{lemma}
\begin{proof}
    Let at time $t$ the configuration is a $j-$LSCC  where $j >1$. In this case if $\mathcal{IN}_j =\phi$ then we have nothing to prove for the first part as $t_1 = t$. So, let at time $t$, $\mathcal{IN}_j \ne \phi$. Now to prove the first part we have to show that there exists a time $t'> t$ such that at time $t'$ the configuration is again a $j-$LSCC  configuration where $|\mathcal{IN}_j|$ at $t' < |\mathcal{IN}_j|$ at $t$ and between $[t,t']$ no robot in $\mathcal{OUT}_j$ at time $t$ moves even if it is activated. Let at $t$, $\mathcal{IN}_j = \{r_{i_1},r_{i_2}, \dots r_{i_p}\}$  where $p \ge 1$. Note that, distance between $r_{i_k}$ and $c$ is strictly less than distance between $r$ and $c$ for any $r \in \mathcal{OUT}_j$ and any $k \in \{1, 2 \dots p \}$ at time $t$. So at $t$, a robot nearest to $c$ must be from $\mathcal{IN}_j$. Let $r_{i_n}$ be one such robot. Then upon activation, it moves left after changing the color to \texttt{done}. Let $t'$ be the first time instance such that for all $t_x \in [t,t'), r_{i_n}$ is on $v_{j}$ in $ \mathcal{C}(t_x)$ and in $\mathcal{C}(t')$, $r_{i_n}$ is on $v_{j+1}$ with color \texttt{done}. So, $|\mathcal{IN}_j|$  at $t'$ $ < |\mathcal{IN}_j|$ at $t$. Now we only have to show that any robot of color \texttt{moving1} in $\mathcal{OUT}_j$ at time $t$ stays on $v_j$ in the time interval $[t,t']$. If possible let some robots that were in $\mathcal{OUT}_j$ at time $t$ move right, on or, before $t'$. Let $r_o$ be the first such robot to move right. Then $r_o$ must have been activated at some time $t'_1 < t'$ when it is nearest to $c$. But since $t_1'<t'$,  in $\mathcal{C}(t_1'), r_{i_n}$ was on $v_j$. Thus in $\mathcal{C}(t_1')$, $r_o$ can ot be nearest to $c$.  So  at $t'$, the configuration is again a $j-$LSCC  configuration where $|\mathcal{IN}_j|$ at $t' < |\mathcal{IN}_j|$ at $t$ and between $[t,t']$ no robot in $\mathcal{OUT}_j$ at time $t$ moves even if it is activated. So, $|\mathcal{OUT}_j|$ at time $t'$ remains same to $   |\mathcal{OUT}_j| $ at time $t$. This proves eventually there is a time $t_1$ when the configuration is a $j-$LSCC with $\mathcal{IN}_j = \phi$ and $\mathcal{OUT}_j$ at time $t_1 =   \mathcal{OUT}_j$ at time $t$.

    Now, for the second part, we have to prove that there exists a time $t_2 > t_1$ when all robots on $v_j$ at $t_1$ are on $v_{j-1}$ at time $t_2$. Thus we have to show that a robot after reaching $v_{j-1}$ from $v_j$ does not do anything until all robots of $v_{j}$ at time $t_1$ reach $v_{j-1}$. If possible let a robot $r$ after reaching $v_{j-1}$ moves again before all robots of $v_j$ move. This implies there exists a time $t_2'> t_1$ when there are robots on $v_j$  with color \texttt{moving1} but $r$ does not see any robot of color \texttt{moving1} on $v_j$ from $v_{j-1}$ at time $t_2'$. This implies there must exist at least one robot, say $r'$, on $v_j$ which does not have color \texttt{moving1} and which obstructs $r$ from seeing any robot of color \texttt{moving1} on $v_j$ at time $t_2'$ (Figure~\ref{fig:emasterLemma}). Note that $r'$ must be of color \texttt{done} and it must have moved from $v_{j-1}$ to $v_j$ after changing its color to \texttt{done} from \texttt{moving1}. This is because at $t_1$, $\mathcal{IN}_j = \phi$, so $r'$ has not changed its color to \texttt{done} on $v_j$. So, at $t_2'$ there must be at least one robot of color \texttt{done}  on $v_j$ which has moved to $v_j$ from $v_{j-1}$ after $t_1$.
    \begin{figure}[h]
    \centering
    \includegraphics[width=4cm]{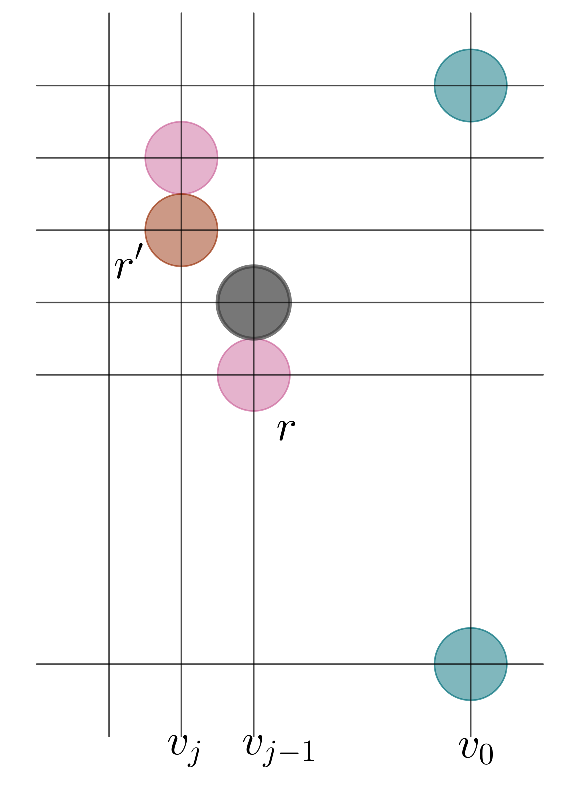}
    \caption{Configuration $\mathcal{C}(t_2')$. Here $r$ on $v_{j-1}$ has no visible robot of color \texttt{moving1} on $v_j$ in spite of $v_j$ having such a robot. $r'$ actually obstructs the view of $r$.}
    \label{fig:emasterLemma}
\end{figure}
    Without loss of generality let $r'$ be the first robot that moves to $v_j$ from $v_{j-1}$ after changing its color to \texttt{done} from \texttt{moving1}. Let $r'$ is activated on $v_{j-1}$ with color \texttt{moving1} at some time $t_4'$ where $t_2'>t_4'>t_1$. This implies $r'$ also has not seen any robot on $v_j$ of color \texttt{moving1} at time $t_4'$. Since  $t_2'>t_4'>t_1$, there must exist a robot of color \texttt{moving1} on $v_j$ at time $t_4'$. $r'$ does not see that robot at $t_4'$ implies there must exists another robot, say $r_s$ of color \texttt{done} on $v_j$ at $t_4'$. Also, $r_s$ must have moved to $v_{j-1}$ to $v_j$ for the same reason described above. So $r'$ can not be the first robot that moved to $v_j$ from $v_{j-1}$ after changing its color to \texttt{done} from \texttt{moving1} after $t_1$. So no robot  that reached $v_{j-1}$ from $v_j$ after $t_1$, moves until all robots of $v_j$ at time $t_1$ moves to $v_{j-1}$. Let $t_2$ be the time when the last robot of $v_j$ reaches $v_{j-1}$ after $t_1$. Note that all robots of color \texttt{moving1}  at time $t_2$ are on $v_{j-1}$ and they are the only ones on $v_{j-1}$. Also, there are no robots between $v_{j-1}$ and $v_0$ and all robots on the left of $v_{j-1}$ have color \texttt{done}. So this configuration at time $t_2$ is a $(j-1)-$LSCC. Also, all robots that were on $v_j$ at time $t_1$ are now on $v_{j-1}$ at time $t_2$ and no other robots moved onto $v_{j-1}$. So, the set of all robots on $v_{j-1}$ at time $t_2 = $ $\mathcal{IN}_{j-1} \cup \mathcal{OUT}_{j-1}$ at time $t_2 = \mathcal{OUT}_j$ at time $t_1$.   
    
\end{proof}
We have seen in the description that, in Phase 2 the configuration becomes a $(\lceil\frac{d}{2}\rceil-1)-$LSCC configuration. Now this lemma tells that for all $j > 1$ and $j \le \lceil\frac{d}{2}\rceil-1$ there is a time when the configuration becomes a $(j-1)-$ SLCC from $j-$SLCC. So starting from $(\lceil\frac{d}{2}\rceil-1)-$LSCC we will eventually have $(\lceil\frac{d}{2}\rceil-2)-$LSCC then $(\lceil\frac{d}{2}\rceil-3)-$LSCC and so on until we have a 1-LSCC. Also note that from $i-$LSCC, $j-$LSCC can not be formed if $j>i$ as robots of color \texttt{moving1} that are on the left of $v_0$  only moves right in Phase 2.
Moreover, we can have the following corollary.
\begin{corollary}
\label{cor: master}
After $(\lceil \frac{d}{2}\rceil-1)$-LSCC is formed, for all $j \in \{1, 2, \dots \lceil \frac{d}{2}\rceil-2\}$ and for any $i <j$, a configuration can not have robots on $v_i$  until $j$-LSCC is formed.
\end{corollary}
\begin{proof}
Let us fix a  $j \in  \{1, 2, \dots \lceil \frac{d}{2}\rceil-2\}$.
    Let at time $t$ which is after $(\lceil \frac{d}{2}\rceil-1)$-LSCC is formed, the configuration is either a $p-$LSCC or, $t$ is between $t_1$ and $t_2$ such that at  $t_1$ the configuration is a $p-$LSCC and at $t_2$ it becomes $(p-1)-$LSCC where $j < p \le \lceil \frac{d}{2}\rceil-1 $ . Now by Lemma~\ref{lemma: master}, at time $t$ there can be no robots on $v_s$ where $s < p-1.$ Now for any $i <j$, we have $i <j< p \implies i < j \le p-1$. Hence at $t$, there can be no robots on $v_i$ where $i <j$.
  \end{proof}
Now we prove the following lemma that ensures that a robot of color  always sees at least one robot of color \texttt{diameter} not on its own vertical line in Phase 2.

\begin{lemma}
\label{lemma:decide Phase 2 moving1}
    In Phase 2, a robot $r$ of color \texttt{moving1} can always see a robot of color \texttt{diameter} that is not on $\mathcal{L}_V(r)$. 
\end{lemma}
\begin{proof}
    Let $t$ be a time when a robot, say $r$, of color \texttt{moving1} on $v_j$ can not see both of $r_1$ and $r_2$, the two robots of color \texttt{diameter} on $v_0$. Without loss of generality let $r$ be on the left of $v_0$. This implies at time $t$, there must exist two robots $r_1'$ and $r_2'$ strictly inside the rectangles bounded by $\mathcal{L}_H(r), \mathcal{L}_H(r_1),v_0,v_j$ and $\mathcal{L}_H(r), \mathcal{L}_H(r_2),v_0,v_j$ respectively (Figure~\ref{fig:Moving1 sees diameter}) i.e., $r_1'$ and $r_2'$ are on $v_{i_1}$ and $v_{i_2}$ where $j > i_1 \ge i_2$. Also, $j > 1$ and $j<\lceil\frac{d}{2}\rceil$ as even if a robot can have color \texttt{moving1} on $v_{\lceil\frac{d}{2}\rceil}$ it performs first look phase as a robot of color \texttt{moving1} after it moves right. Also from the description $(\lceil\frac{d}{2}\rceil-1)-$LSCC  is formed first until then no robot moves to $v_i$ where $i< \lceil\frac{d}{2}\rceil-1$. So until $(\lceil\frac{d}{2}\rceil-1)-$LSCC  is formed all robots of color \texttt{moving1} must see both the robots $r_1$ and $r_2$ from $v_{\lceil\frac{d}{2}\rceil-1}$. So $t$ must be a time after $(\lceil\frac{d}{2}\rceil-1)-$LSCC  is formed. Observe that, $r_1'$ and $r_2'$ must be of color \texttt{moving1}. Now since at time $t$, there are robots on $v_{i_1}$ and $v_{i_2}$ where $i_1, i_2 < j$, there exists a time $t' < t$ when the configuration was a $j-$LSCC such that $\mathcal{IN}_j = \phi$ (Lemma~\ref{lemma: master} and Corollary~\ref{cor: master}) (Figure~\ref{fig:Moving1 sees diameter}). Also in $\mathcal{C}(t)'$, $r, r_1', r_2'$ were on $v_j$. Now similar to lemma~\ref{lemma:off decide phase 2}, let us denote the half where $r_1'$ is located as the upper half and the other one as the lower half. Now the first robot that moves to $v_{j-1}$ from $v_j$ must be nearest to $c$ at time $t'$. Without loss of generality let it be on the upper half at time $t'$. There can be at most another such robot which is also nearest to $c$ at time $t'$. If another such robot exists then it also has to be on the upper half at time $t'$ otherwise $r$ becomes nearer to $c$ at time $t$ and moves to $v_{j-1}$ before any other robots and thus at $t$ $r$ can not be at $v_j$ as assumed. Also at time $t$, there is at least one robot of color \texttt{moving1} on the lower half. So let $r_l$ be the first robot from the lower half at time $t'$ that has moved to  $v_{j-1}$ from $v_j$. Then it must have moved from $v_j$ after seeing a robot, say $r_l'$, of color \texttt{moving1} on $\mathcal{R}_I(r_l) = v_{j-1}$ such that there is no other robots except $r_l$ and $r_l'$ on or between the rectangle, say $R$, bounded by $v_j, v_{j-1}, \mathcal{L}_H(r_l)$ and $\mathcal{L}_H(r_l')$ at a time $t_1' >t'$ and $t> t_1'$. Now since $r_l$ is the first robot from the lower half to move after $t'$, $r_l'$ must be from the upper half. Thus at time $t_1'$, $R$ must contain $r$ other than $r_l$ and $r_l'$  which is a contradiction. Hence, a robot of color \texttt{moving1} always sees a robot of color \texttt{diameter}. And according to the algorithm, since $r$ never reaches $v_0$ in Phase 2, it can not see $r_1$ and $r_2$ on its own vertical line. 
\end{proof}

\begin{figure}[h]
    \centering
    \includegraphics[width=8cm]{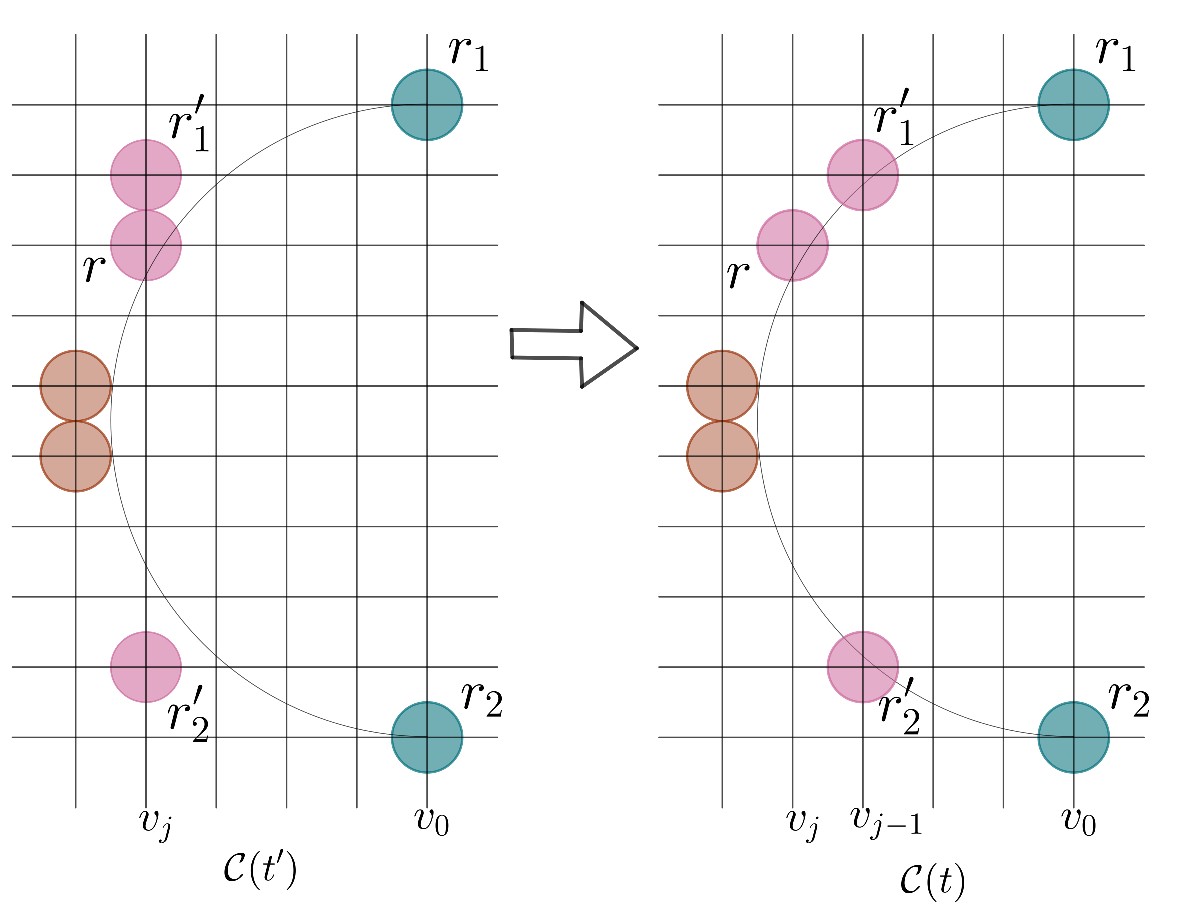}
    \caption{$\mathcal{C}(t')$ is a $j-LSCC$ where all robots on $v_j$ are in $\mathcal{OUT}_j$. From $\mathcal{C}(t')$, $\mathcal{C}(t)$ is formed where $r$ is still on $v_j$ but $r_1'$ and $r_2'$ are on right of $v_j$ obstructing $r$ from seeing $r_1$ and $r_2$. Here $r_1$ and $r_2$ are the robots of color \texttt{diameter} on $v_0$.}
    \label{fig:Moving1 sees diameter}
\end{figure}

Let $r$ be a robot in Phase 2 on the left of $v_0$. Note that $\mathcal{L}_H(r)$ intersects the circle $\mathcal{CIR}$ exactly once at a point $C_r$ on the left of $v_0$. Thus, there exists exactly one grid point denoted as, $C_T(r)$ on $\mathcal{L}_H(r)$ such that either $ dist(C_r,C_T(r))=0$ or $C_T(r)$ is on the left of $C_r$ such that $dist(C_r, C_T(r)) <1 $ . Then we define $C_T(r)$ to be the \textit{Terminating Point} of $r$. 
 We can now have the following observation
 \begin{observation}
 \label{obs: all robots on v1 on circle}
     In a 1-LSCC if a robot $r \in \mathcal{OUT}_1$, then $r$ must be on $C_T(r)$.
 \end{observation}
For $r$, $C_T(r) = \mathcal{L}_H(r) \cap v_j$ for some $j \ge 1$. Now, in a 1-LSCC if $r \in \mathcal{OUT}_1$, then $j 
\le 1$. Hence, $j=1$.

We now first ensure that if a robot $r$ terminates, it does not terminate on a grid point that is not on $C_T(r)$.
\begin{lemma}
    A robot $r$ can only terminate on $C_T(r)$.
    \label{lemma: terminating grid point}
\end{lemma}
\begin{proof}
    For this, first note that a robot can only terminate after $(\lceil\frac{d}{2}\rceil-1)-$LSCC is formed. Now let for a robot $r$, $C_T(r)$ is the grid point $\mathcal{L}_H(r) \cap v_j$ for some $j \in [1, \lceil\frac{d}{2}\rceil] \cap \mathbb{N}$. If $j =  \lceil\frac{d}{2}\rceil$ then in  $(\lceil\frac{d}{2}\rceil-1)-$LSCC, $r \in \mathcal{IN}_{\lceil\frac{d}{2}\rceil-1}$. Then eventually $r$ moves to $v_{\lceil\frac{d}{2}\rceil}$ after changing the color to \texttt{done}. When it is activated next it terminates on $C_T(r) = v_{\lceil\frac{d}{2}\rceil} \cap \mathcal{L}_H(r)$.
So if $j= \lceil\frac{d}{2}\rceil$ then $r$ can not terminate on any other position except $C_T(r)$.Now let us consider the case $j < \lceil\frac{d}{2}\rceil$.  If possible let $r$ terminates on $\mathcal{L}_H(r) \cap v_i$ where $i > j \ge 1$. Then there must exist a time when $r$ is on $v_{i-1}$. Then by Corollary~\ref{cor: master} there exists a time when the configuration is $i-$LSCC. Here $r \in \mathcal{OUT}_i$ and it is strictly outside the circle $\mathcal{CIR}$ as $i>j$. So, again by lemma~\ref{lemma: master} there exists a time when the configuration becomes a $(i-1)-$LSCC and $r$ is on $v_{i-1}$. Here $i-1 \ge j$ i.e $r \in \mathcal{OUT}_{i-1}$. Now if $ i-1 = 1$ then $r$ eventually changes the color to \texttt{done} and terminates on $v_j= v_1$ as  $C_T(r) = \mathcal{L}_H(r) \cap v_1$ for this case (Observation~\ref{obs: all robots on v1 on circle}). This is contrary to our assumption that $r$ terminates on $v_i$ where $i > j =1$. So, let $i-1 > 1$. Then eventually $(i-2)-$LSCC will be formed and $r$ will be on $v_{i-2}$. Now note that according to the algorithm for Phase 2, after $(\lceil\frac{d}{2}\rceil-1)-$LSCC is formed, only a robot of color \texttt{moving1} can move further from $v_0$ after changing the color to \texttt{done}. So even if a robot moves further from $v_0$ it can move in such a way only once as after that it terminates. So if $r$ with color \texttt{moving1} is on $v_{i-2}$ at some time in Phase 2, it can not move back to $v_{i}$ and terminate.  Thus we reach a contradiction assuming $r$ terminates on $v_i$ where $i > j$ where $C_T(r) = \mathcal{L}_H(r) \cap v_j $ for some $j \in [1, \lceil\frac{d}{2}\rceil] \cap \mathbb{N}$. Thus $r$ must terminate either on $C_T(r) = \mathcal{L}_H(r) \cap v_j$ or on $\mathcal{L}_H(r) \cap v_i$ where $i < j$. If possible let $r$ terminates on $\mathcal{L}_H(r) \cap v_i$ where $i < j$. This implies there exists a time $t$ when the configuration is a $j-$LSCC. In this configuration $r$ is on $C_T(r)$ thus $r \in \mathcal{OUT}_j$. Thus eventually $(j-1)-$LSCC will be formed where $r$ is on $v_{j-1}$. Note that in this configuration $r \in \mathcal{IN}_{j-1}$. So, eventually, $r$ will change its color to \texttt{done} and move to $v_j$. So $r$ terminates on $\mathcal{L}_H(r) \cap v_j = C_T(r)$ contrary to our assumption. Hence if $r$ terminates it must terminate at $C_T(r)$.
\end{proof}
Now we will prove that all robots that are not on $v_0$ terminates
\begin{lemma}
    \label{lemma: all robot terminates}
    All robot that are not on $v_0$ terminates eventually during Phase 2. 
\end{lemma}
\begin{proof}
    If possible let  $r$ be a robot on the left of $v_0$ that never terminates in Phase 2. Let $C_T(r) = \mathcal{L}_H(r) \cap v_j$ for some $j \in \{1,2, \dots \lceil \frac{d}{2}\rceil\}$.
    For $j>1$, $r$  can terminate after it moves to $C_T(r)$ from $v_{j-1}$ of a $(j-1)-$LSCC after changing the color to \texttt{done}. Now as it is assumed that $r$ does not terminate, either $(j-1)-$LSCC is never formed or, even if it is formed $r$ is not on $v_{j-1}$ in $(j-1)-$LSCC. Now by Lemma~\ref{lemma: master} if there is at least one robot that has not terminated $(j-1)-$LSCC  will be formed eventually. Thus if $r$ never terminates there can be only one possibility that when $(j-1)-$LSCC is formed $r$ is not there on $v_{j-1}$. This implies $r$ must be on some $v_i$ where $i > j$ with color \texttt{done}. This implies $r$ terminates at $v_i \cap \mathcal{L}_H(r)$ for some $i > j$. This is impossible due to Lemma~\ref{lemma: terminating grid point}. Hence $r$ must terminate if $j> 1$. Now similarly, for $j=1$, $r$ never terminates implies, When $1-$LSCC is  formed, $r$ is not on $v_1$. Again this is impossible due to similar reasons as above.
    Thus all robots that are not on $v_0$ must terminate in Phase 2
\end{proof}
 Now using Lemma~\ref{lemma: terminating grid point} and Lemma~\ref{lemma: all robot terminates} and the fact that the robots of color \texttt{diameter} terminates after changing its color to \texttt{diameter} either from \texttt{chord} or from \texttt{moving1} we can state the following theorem.
 \begin{theorem}
 \label{thm: final}
     From any initial configuration within finite time, all opaque fat robots with one axis agreement can terminate after forming a circle on an infinite grid under asynchronous scheduler by executing the algorithm $CIRCLE\_FG$.  
 \end{theorem}
\section{Conclusion}
The problem of circle formation is widely studied
in the field of swarm robotics. It has been studied under various assumptions on plane. But in the discrete domain, the work is limited. With obstructed visibility
model, this problem has been considered on the plane and infinite grid using luminous
opaque robots. But using fat robots (i.e. robots with certain dimensions), it was only
done in a plane. In this paper, we have taken care of this. We have shown that with a
swarm of luminous opaque fat robots having one-axis agreement on an infinite grid,
a circle of diameter $O(n)$ ($n$ is the number of robots in the system) can be formed from any initial configuration using one light having
5 distinct colors. For future
courses of research, one way is to find out the optimal number of colors needed to solve this problem. Also, it would be interesting to see if the same problem can be solved
using disoriented robots.

\textbf{Acknowledgements:} The  First author is supported by UGC [1132/(CSIR-UGC NET
JUNE 2017)], the Government
of India and the third author is supported by the West Bengal State government
Fellowship Scheme e [P-1/RS/23/2020].

\textbf{Disclosure Statement:} The authors report there are no competing interests to declare.
\bibliographystyle{tfnlm}
\bibliography{interactnlmsample}

\end{document}